\documentclass{fundam}
\usepackage{amssymb}
\usepackage{extarrows}
\usepackage{tikz}
\usetikzlibrary{graphs,quotes,petri}
\usepackage[ruled, linesnumbered, vlined]{algorithm2e}
\usepackage{csquotes}
\usepackage{caption}
\usepackage[misc,geometry]{ifsym} 
\usepackage{enumitem} 
\usepackage{float}
\usepackage{tabularx,lipsum,environ}
\usepackage{lineno}

\makeatletter
\newcommand{\problemtitle}[1]{\gdef\@problemtitle{#1}}
\newcommand{\probleminput}[1]{\gdef\@probleminput{#1}}
\newcommand{\problemquestion}[1]{\gdef\@problemquestion{#1}}
\NewEnviron{decisionproblem}{
  \problemtitle{}\probleminput{}\problemquestion{}
  \BODY
  \par\addvspace{.5\baselineskip}
  \noindent
  \begin{tabularx}{0.97\textwidth}{@{\hspace{\parindent}} l X c}
    \multicolumn{2}{@{\hspace{\parindent}}l}{\@problemtitle} \\
    \emph{Input:} & \@probleminput \\
    \emph{Question:} & \@problemquestion
  \end{tabularx}
  \par\addvspace{.5\baselineskip}
}
\makeatother

\newcommand{\drop}[1]{}

\newcommand{\escale}[1]{\ensuremath{\textbf{\scalebox{0.8}{#1}}}}

\newcommand{\nscale}[1]{\ensuremath{\textbf{\scalebox{0.8}{#1}}}}

\newcommand{\myEdge}[2]{ \tikz[baseline=-1pt]{
\draw[#2,line width=0.3pt] (0,0) -- ++(0.6,0) node[anchor=base, yshift=3pt, pos=0.5] {\escale{$#1$}};
}}
\newcommand{\mylEdge}[2]{ \tikz[baseline=-1pt]{
\draw[#2,line width=0.3pt] (0,0) -- ++(0.9,0) node[anchor=base, yshift=4pt, pos=0.5] {\escale{$#1$}};
}}
\newcommand{\myLEdge}[2]{ \tikz[baseline=-1pt]{
\draw[#2,line width=0.3pt] (0,0) -- ++(1.2,0) node[anchor=base, yshift=4pt, pos=0.5] {\escale{$#1$}};
}}
\newcommand{\Edge}[1]{ \tikz[baseline=-1pt]{
\draw[->,line width=0.3pt] (0,0) -- ++(0.6,0) node[anchor=base, yshift=4.5pt, pos=0.5] {\escale{$#1$}};
}}

\newcommand{\lEdge}[1]{ \tikz[baseline=-1pt]{
\draw[->,line width=0.3pt] (0,0) -- ++(0.88,0) node[anchor=base, yshift=5pt, pos=0.5] {\escale{$#1$}};
}}
\newcommand{\ledge}[1]{ \tikz[baseline=-1pt]{
\draw[->,line width=0.3pt] (0,0) -- ++(0.8,0) node[anchor=base, yshift=3pt, pos=0.5] {\escale{$#1$}};
}}

\newcommand{\edge}[1]{\myEdge{#1}{->}}

\newcommand{\fbedge}[1]{\myEdge{#1}{<->}}
\newcommand{\fBedge}[1]{\mylEdge{#1}{<->}}
\newcommand{\FBedge}[1]{\myLEdge{#1}{<->}}
\newcommand{\nop}{\ensuremath{\textsf{nop}}}
\newcommand{\inp}{\ensuremath{\textsf{inp}}}
\newcommand{\out}{\ensuremath{\textsf{out}}}
\newcommand{\set}{\ensuremath{\textsf{set}}}
\newcommand{\res}{\ensuremath{\textsf{res}}}
\newcommand{\swap}{\ensuremath{\textsf{swap}}}
\newcommand{\free}{\ensuremath{\textsf{free}}}
\newcommand{\used}{\ensuremath{\textsf{used}}}

\newcommand{\K}{\ensuremath{\mathfrak{K}}}
\newcommand{\E}{\ensuremath{\mathfrak{E}}}
\newcommand{\mS}{\ensuremath{\mathfrak{S}}}
\newcommand{\mL}{\ensuremath{\mathfrak{L}}}
\newcommand{\U}{\ensuremath{\mathfrak{U}}}
\newcommand{\Ss}{\ensuremath{\mathcal{S}}}
\newcommand{\R}{\ensuremath{\mathcal{R}}}

\tikzstyle{place}=[circle,thick,draw=blue!75,fill=blue!20,minimum size=6mm]
\tikzstyle{red place}=[place,draw=red!75,fill=red!20]
\tikzstyle{transition}=[rectangle,thick,draw=black!75,
  			  fill=black!20,minimum size=4mm]
			   \tikzstyle{every label}=[red]

\begin{document}

\setcounter{page}{261}
\publyear{22}
\papernumber{2161}
\volume{189}
\issue{3-4}

  \finalVersionForARXIV

\title{On the Complexity of Techniques That Make Transition Systems Implementable by Boolean Nets}

\author{Raymond Devillers
 \\
D\'epartement d'Informatique \\
Universit\'e Libre de Bruxelles\\
Boulevard du Triomphe\\
C.P. 212, B-1050 Bruxelles, Belgium\\
 raymond.devillersl@ulb.be
\and
 Ronny Tredup\thanks{Address  for correspondence: Geschwister-Scholl-Gymnasium L\"obau,
                                   Pestalozzistra\ss e 21, 02708 L\"obau Germany}
 \\
 Geschwister-Scholl-Gymnasium L\"obau\\
 Pestalozzistra\ss e 21, 02708 L\"obau Germany\\
 tredupron@gsg-loebau.lernsax.de
}

\maketitle

\runninghead{R. Devillers and R. Tredup}{The Complexity of Techniques That Make Transition Systems Implementable...}

\begin{abstract}
Let us consider some class of (Petri) nets.
The corresponding \emph{Synthesis} problem consists in deciding whether a given labeled transition system (TS) $A$ can be implemented by a net $N$ of that class.
In case of a negative decision, it may be possible to convert $A$ into an implementable TS $B$ by
applying various modification techniques, like
relabeling edges that previously had the same label, suppressing edges/states/events, etc.
It may however be useful to limit the number of such modifications to stay close to the original problem, or optimize the technique.
In this paper, we show that most of the corresponding problems are NP-complete if the considered class corresponds to
so-called flip-flop nets or some flip-flop net derivatives.
\end{abstract}

\begin{keywords}
Petri net, Boolean Types, Label-Splitting, Edge-Removal, Event-Removal, State-Removal, Synthesis, Complexity\medskip
\end{keywords}

\section{Introduction}%

The so-called synthesis problem for nets of some class $\tau$ ($\tau$-synthesis, for short) consists in deciding whether for a given
labeled transition system (TS, for short) $A$ there is a net $N$ of class $\tau$ (a $\tau$-net, for short) that implements $A$.
In case of a positive decision, $N$ should be constructed, possibly while minimizing some features (number of places, arc weights, \ldots); in case of a negative decision, some reason(s) should be given.

$\tau$-Synthesis is used to, for example, extract concurrency from sequential specifications like TS and languages~\cite{fac/BadouelCD02} and has applications in, for example, process discovery~\cite{daglib/0027363}, supervisory control~\cite{deds/HollowayKG97} and the synthesis of speed independent circuits~\cite{tcad/CortadellaKKLY97}.

However, whether $N$ exists depends crucially on the kind of implementation we are striving at, that is, whether $N$ should be an (exact) \emph{realization} (meaning that $A$ and $N$'s reachability graph are isomorphic), or a \emph{language-simulation}
(meaning that $A$ and $N$ have the same language) or an \emph{embedding}
(meaning that $N$ preserves the distinctness of states of $A$).
Unfortunately, whatever the kind of implementation, a solution does not always exist.
This observation motivates the search for techniques that modify the given TS (as little as possible) so that the result is an implementable behavior.

For instance, \emph{label-splitting} has been considered in \cite{txtcs/BadouelBD15,topnoc/Carmona12,707587,topnoc/SchlachterW19}:
This approach may convert a non-implementable TS $A$ into an implementable one $A'$ by relabeling differently some edges that previously had the same label.
However, the new events produced by the label-splitting increase the complexity of the derived net, since each new copy will be transformed into a new transition.
Hence, it is desired to find a label-splitting that induces the minimal number of transitions in the sought net.
This allows to consider $\tau$-\emph{label-splitting} as a decision problem with input $A$ and $\kappa\in \mathbb{N}$;
the question is whether there is a TS $B$  implementable by a $\tau$-net that, firstly, is derived from $A$ by splitting labels and, secondly, has at most $\kappa$ labels (then, by a dichotomic search, an optimal splitting may be found).
Recently, in~\cite{corr/abs-2002-04841}, it has been shown that $\tau$-label-splitting aiming at embedding is NP-complete if $\tau$ equals the
class of \emph{weighted Place/Transition}-nets.
Moreover, in~\cite{ictcs/Tredup20}, we have shown that $\tau$-label-splitting aiming at language-simulation or realization is also NP-complete for this type.

Since label-splitting is intractable (at least for weighted Place/Transition nets), other techniques with better worst-case complexity would be preferable.
This led to the idea of simplifying an unimplementable TS by removing some edges, events or states until the result is implementable.
Indeed, the removal of components is a strong technique that always leads to implementable behavior, since the result is implementable when only the initial state is left at the end.
However, such an extreme modification is certainly not desirable.
Instead, we aim to remove the minimum number $\kappa$ of components so that the result is implementable.
This justifies to consider the decision problems \emph{edge-, event-, and state-removal}:
given a TS $A$, and a number $\kappa$, the task is to decide whether we can modify $A$ to an implementable TS $B$ by the removal of at most $\kappa$ edges, events or states, respectively.
Unfortunately, it has been shown that these removal techniques are also NP-complete for all implementations - embedding,
language-simulation and realization - if we target (weighted) Place/Transition \linebreak nets~\cite{corr/abs-2112-03605,apn/Tredup21}.

Naturally, it raises the question whether the decision problems, and thus the corresponding modification techniques, are of a different complexity provided the net class sought is different.

\eject

A whole family of net classes for which such investigations are certainly of interest is defined by
the so-called \emph{Boolean types of nets}~\cite{txtcs/BadouelBD15,acta/BadouelD95,acta/KleijnKPR13,acta/MontanariR95,apn/Pietkiewicz-Koutny97,ac/RozenbergE96,stacs/Schmitt96}, since the respective nets are widely accepted as excellent tools for modeling concurrent and distributed systems.

Boolean nets allow at most one token on each place $p$ in every reachable marking.
Therefore, $p$ is considered a Boolean condition that is \emph{true} if $p$ is marked and \emph{false} otherwise.
A place $p$ and a transition $t$ of a Boolean net $N$ are related by one of the following Boolean \emph{interactions}
(partial functions):
\emph{no operation} (\nop), \emph{input} (\inp), \emph{output} (\out), \emph{unconditionally set to true} (\set), \emph{unconditionally
reset to false} (\res), \emph{inverting} (\swap), \emph{test if true} (\used), and \emph{test if false} (\free).
The relation between $p$ and $t$ determines which conditions $p$ must satisfy to allow $t$'s firing (the corresponding partial function must be defined), and which impact has the firing of $t$ on $p$ (the partial function determines what should be the new value of $p$ if $t$ is executed).
Boolean nets are then classified by the interactions of $I=\{\nop,\inp,\out,\res,\set,\swap,\used,\free\}$ that they apply or spare.
More exactly, a subset $\tau\subseteq I$ is called a \emph{Boolean type of net} and a net $N$ is of type $\tau$ (a $\tau$-net) if it applies at most the interactions of $\tau$.

However, the determination of the complexity of TS modifications allowing $\tau$-synthesis does not actually define an open problem for all Boolean types of nets.
In particular, it is known that $\tau$-synthesis aiming at language-simulation or realization is NP-complete for 84 out of the 128 Boolean types of nets containing \nop{} (which allows some kind of independence between places and transitions)~\cite{tcs/BadouelBD97,topnoc/Tredup21, tamc/TredupR19,apn/TredupR19,topnoc/Tredup21}.
Likewise, it is known that $\tau$-synthesis striving at embedding is NP-complete for 90 of the \nop-equipped Boolean types~\cite{ictac/TredupE20}.
Consequently, for these types, the NP-completeness of $\tau$-synthesis implies already the NP-completeness of the corresponding problems \emph{label-splitting}, and \emph{edge-, event-}, and \emph{state-removal}, for instance by choosing $\kappa$ in a way that forbids any splitting, or removal, respectively.
This puts the Boolean types with a tractable synthesis problem in focus.
One of the most prominent Boolean type that fulfills this criterion is the so-called \emph{flip-flop} type of nets,
where $\tau=\{\nop,\inp,\out,\swap\}$.
Flip-flop nets have been originally introduced in~\cite{stacs/Schmitt96}, and their name is inspired by the interaction \swap\ that allows a transition to unconditionally change the current marking of a place from 0 to 1 and from 1 to 0.
The flip-flop nets are considered as the Boolean counterpart of the Place/Transition nets.
This characterization is mainly based on the fact that synthesis aiming at flip-flop nets is solvable by a polynomial time algorithm~\cite{stacs/Schmitt96} that is derived from the algorithm for synthesis aiming at Place/Transition nets, which has been introduced in~\cite{tapsoft/BadouelBD95}.
In fact, the algorithm for flip-flop nets~\cite{stacs/Schmitt96} is extendable to all types $\tau=\{\nop,\swap\}\cup\omega$ with $\omega\subseteq \{\inp,\out,\used,\free\}$~\cite{tamc/TredupR19}, which makes their synthesis problem also tractable.

In this paper, for all 16 types $\tau=\{\nop,\swap\}\cup\omega$ with $\omega \subseteq\{\inp,\out,\used,\free\}$, hence in particular for the flip-flop nets, we investigate the computational complexity of $\tau$-label-splitting and element removing for all introduced implementations: embedding, language-simulation and realization.
In particular, we show that this problem aiming at embedding is NP-complete for all these types, unfortunately.
Moreover, label-splitting aiming at language-simulation or realization is NP-complete if $\omega\not=\emptyset$, otherwise it is tractable.

Furthermore, for all 15 types $\tau=\{\nop,\swap\}\cup\omega$ with $\omega \subseteq\{\inp,\out,\used,\free\}$ and $\omega\not=\emptyset$, we investigate the computational complexity of $\tau$-edge-, event-, and state-removal, and show that these problems are NP-complete for all these types, regardless which of the implementations embedding, language-simulation and realization we are aiming at.

We obtain our NP-completeness results by reductions from a variant of the well-known vertex cover problem.
Our current approach generalizes our methods from~\cite{ictcs/Tredup20,apn/Tredup21} and tailors them to flip-flop nets and its aforementioned derivatives.

This paper, which is an extended version of~\cite{rp/Tredup20}, is organized as follows.
The next Section~\ref{sec:prelis} introduces necessary notions and definitions.
After that, Sections~\ref{sec:label_splitting} to \ref{sec:state_removal}  present
our complexity results for label-splitting and the edge/event/state-removals.
Finally, Section~\ref{sec:con} briefly closes the paper.

\section{Preliminaries}\label{sec:prelis} %

This section introduces the basic notions used throughout the paper.

\begin{definition}[Transition Systems]\label{def:ts}
A \emph{transition system} (TS) $A=(S,E, \delta)$ is a finite directed labeled graph with the set of nodes $S$ (called {\em states}), the set of labels $E$ (called {\em events}) and partial \emph{transition function} $\delta: S\times E \longrightarrow S$.

\medskip
We shall assume no event is useless, i.e., for every $e\in E$ there are states $s,s'\in S$ such that $\delta(s,e)=s'$.
For convenience, with a little abuse of notation, we often identify $\delta$ and the set $\{(s,e,s')\in S\times E\times S\mid \delta(s,e)=s'\}$.

Event $e$ \emph{occurs} at $s$, denoted by $s\edge{e}$, if $\delta(s,e)$ is defined, otherwise we denote it by $s\edge{\neg e}$ or $\neg s\edge{e}$.
We denote $\delta(s,e)=s'$ by $s\edge{e}s'$.

An \emph{initialized} TS $A=(S,E,\delta, \iota)$ is a TS with a distinct {\em initial} state $\iota\in S$ such that
every state $s\in S$ is \emph{reachable} from $\iota$ by a directed labeled path.
If $w=e_1\dots e_n\in E^*$, by $\iota\edge{w}$ we denote that there are states $\iota=s_0,\dots, s_n\in S$ such that $s_i\edge{e_{i+1}}s_{i+1}\in A$ for all $i\in \{0,\dots, n-1\}$, in which case we also denote $\iota\edge{w}s_n$.
We extend this notation by stating that $\iota\edge{\epsilon}\iota$.

The \emph{language} of $A$ is defined by $L(A)=\{ w \in E^*\mid \iota\edge{w} \}$.

When not given explicitly, we shall refer to the components of a TS $A$ by $S(A)$ (states), $E(A)$ (events), $\delta_A$ (transition function) and $\iota_A$ (initial state).
\end{definition}

\begin{definition}[Simulations]
Let $A$ and $B$ be initialized TS with the same set of events $E$.
We say $B$ \emph{simulates} $A$, if there is a mapping $\varphi : S(A)\rightarrow S(B)$ such that $\varphi(\iota_A)=\iota_B$
and $s\edge{e}s'\in A$ implies $\varphi(s)\edge{e}\varphi(s')\in B$;
such a mapping is called a \emph{simulation} (between $A$ and $B$).

\medskip
$\varphi$ is an \emph{embedding}, denoted by $A\hookrightarrow B$, if it is injective;
$\varphi$ is a \emph{language-simulation}, denoted by $A\triangleright B$, if $\varphi(s)\edge{e}$ implies $s\edge{e}$, implying $L(A)=L(B)$~\cite[p.~67]{txtcs/BadouelBD15};
$\varphi$ is an \emph{isomorphism}, denoted by $A\cong B$, if it is both an embedding and a language simulation.
\end{definition}

\begin{definition}[Boolean Types~\cite{txtcs/BadouelBD15}]
A Boolean type is a subset $\tau$ of the 8 Boolean interactions
\[I=\{\nop, \inp, \out, \set, \res, \swap, \used, \free\}\]
\eject

\noindent  i.e., the partial functions $\{0,1\}\to\{0,1\}$ schematized in Figure~\ref{fig:interactions}.
To each type $\tau$, we associate the Boolean\footnote{Meaning that the state set $S=\{0,1\}$. The initial state is irrelevant here.}
TS $A_\tau=(\{0,1\},E_\tau,\delta_\tau)$, where $E_\tau=\tau$ and, for each $x\in \{0,1\}$ and $i\in E_\tau$, $\delta_\tau(x,i)$ is defined if so is $i(x)$, in which case $\delta_\tau(x,i)=i(x)$.
Since the correspondence is unique, we often  identify $\tau$ with $A_\tau$.
\end{definition}
In fact there are $(2+1)^2=9$ such partial functions, but the empty one (defined nowhere) is of no interest here, since it may never occur as the label of an event.

\begin{figure}[!h]
\vspace*{3mm}
\begin{minipage}{1.0\textwidth}
\centering
\scalebox{0.95}{
\begin{tabular}{c|c|c|c|c|c|c|c|c}
$x$ & $\nop(x)$ & $\inp(x)$ & $\out(x)$ & $\set(x)$ & $\res(x)$ & $\swap(x)$ & $\used(x)$ & $\free(x)$\\ \hline
$0$ & $0$ & & $1$ & $1$ & $0$ & $1$ & & $0$\\
$1$ & $1$ & $0$ & & $1$ & $0$ & $0$ & $1$ & \\
\end{tabular} }
\caption{
All interactions $i$ of $I$.
If a cell is empty, then $i$ is undefined on the respective $x$.
}\label{fig:interactions}
\end{minipage}
\vspace{11pt}

\begin{minipage}{\textwidth}
\begin{center}
\begin{tikzpicture}[new set = import nodes]
\begin{scope}
\node (0) at (0,0) {\nscale{$0$}};
\node (1) at (2,0) {\nscale{$1$}};

\path (0) edge [->, out=-120,in=120,looseness=5] node[left] {\escale{$\nop$}} (0);
\path (1) edge [<-, out=60,in=-60,looseness=5] node[right] {\escale{$\nop$}  } (1);

\path (0) edge [<-, bend right= 30] node[below] {\escale{$\inp, \swap$}} (1);
\path (0) edge [->, bend  left = 30] node[above] {\escale{$\swap$}} (1);
\node () at (1,-1.5) {$A_\tau$};

\end{scope}
\begin{scope}[xshift=6.4cm, node distance=0.5cm,bend angle=45,auto]
\node (x) at (0,0) {};
\node [place] (R1)[tokens =1,left of=x, label=above: \nscale{$R_1$}, node distance =1.75cm]{};
\node [transition] (a)[above of =x]{$a$}
     	 edge [above,] node {$\nscale{\inp}$}  (R1);
\node [transition] (b)[below of =x]{$a'$}
     	 edge [ below] node {$\nscale{\nop}$}  (R1);
	
\node [place] (R2)[right of=x, label=above: \nscale{$R_2$}, node distance =1.75cm]{}
	edge [above] node {$\nscale{\swap}$}  (a)
	edge [below] node {$\nscale{\inp}$}  (b);
	\node () at (0,-1.5) {$N$};
\end{scope}
\begin{scope}[xshift=10cm,nodes={set=import nodes}]
		\node (A) at (1.3,-1.5) {$A_N$};
		\node (10) at (0,0) {\nscale{$(1,0)$}};
		\node (01) at (1.75,0) {\nscale{$(0,1)$}};
		\node (00) at (3.5,0) {\nscale{$(0,0)$}};
		\draw[-latex](0,-0.6)to node[]{}(10);
\graph {(import nodes);
			10 ->["\escale{$a$}"]01;
			01 ->["\escale{$a'$}"]00;
			
		};
\end{scope}
\end{tikzpicture}
\end{center} \vspace*{-5.5mm}
\caption{Left: $A_\tau$ for the the type $\tau=\{\nop,\inp,\swap\}$.
Middle: A $\tau$-net $N$ (as usual, the initial marking is indicated by putting a token in the places $p$ such that $M_0(p)=1$).
Right: The reachability graph $A_N$ of $N$, where each state $M$ is represented by its marking $(M(R_1),M(R_2))$,
and the initial state is indicated by the arrow without source state.}\label{fig:type_net_rg}
\end{minipage}\vspace*{-3mm}
\end{figure}


\begin{definition}[$\tau$-Nets]
Let $\tau\subseteq I$.
A Boolean net $N = (P, T, f, M_0)$ of type $\tau$ (a {\em $\tau$-net}) is given by finite disjoint sets $P$ of {\em places} and $T$ of {\em transitions}, a (total) {\em flow-function} $f: P \times T \rightarrow \tau$, and an {\em initial marking} $M_0: P\longrightarrow  \{0,1\}$.
A transition $t \in T$ can {\em fire} in a marking $M: P\longrightarrow  \{0,1\}$ if $\delta_\tau(M(p), f(p,t))$ is defined for all $p\in P$.
By firing, $t$ produces the marking $M' : P\longrightarrow  \{0,1\}$ where $M'(p)=\delta_\tau(M(p), f(p,t))$ for all $p\in P$, denoted by $M \edge{t} M'$.
The behavior of $\tau$-net $N$ is captured by an initialized transition system $A_N$, called the {\em reachability graph} of $N$.
The states set $S(A_N)$ of $A_N$ will also be denoted  $RS(N)$;
 it consists of all markings that can be reached from the initial state $M_0$ by sequences of transition firings;
$E(A_N)=T$, $\iota_{A_N}=M_0$, and $(s,e,s')\in\delta_{A_N}$ iff $s\edge{e}s'$ in $N$.
\end{definition}

\begin{example}
Figure~\ref{fig:type_net_rg} shows the transition system $A_\tau$ for the type $\tau=\{\nop,\inp,\swap\}$, the $\tau$-net
$N=(\{R_1,R_2\}, \{a,a'\}, f, M_0)$ with places $R_1, R_2$, flow-function $f(R_1,a)=f(R_2, a')=\inp$, $f(R_1,a')=\nop$, $f(R_2,a)=\swap$ and initial marking $M_0$ defined by $(M_0(R_1),M_0(R_2))=(1,0)$.
Since $1\edge{\inp}0\in \tau$ and $0\edge{\swap}1\in \tau$, the transition $a$ can fire in $M_0$, which leads to the marking $M=(M(R_1), M(R_2))=(0,1)$.
After that, $a'$ can fire, which results in the marking \mbox{$M'=(M'(R_1), M'(R_2))=(0,0)$.}
The reachability graph $A_N$ of $N$ is depicted on the right hand side of Figure~\ref{fig:type_net_rg}.
\end{example}

\begin{definition}[Implementations]
Let $A$ be an initialized TS, $\tau$ be a Boolean type and $N$ a $\tau$-net.
We say $N$ is an (exact) \emph{realization} of $A$ if $A\cong A_N$.
If $A\triangleright A_N$, then $N$ is a \emph{language-simulation} of $A$.
If $A\hookrightarrow A_N$, then $N$ is an embedding of $A$.
\end{definition}

\begin{remark}\label{RemInit}
By definition, the reachability graph of a net is initialized.
In the following, we shall thus only consider initialized transition systems without explicitly mention it every time.
\end{remark}

Let $\tau$ be a Boolean type of nets.
If a TS $A$ is implementable by a $\tau$-net $N$, then we want to construct $N$ from the structure of $A$.
Since $A_N$ has to simulate $A$, $N$'s transitions correspond to $A$'s events.
The connection between global states in TS and local states in the sought net is given by {\em regions of TS} that mimic places:

\begin{definition}[$\tau$-Regions]
A $\tau$-region $R=(sup, sig)$ of $A=(S, E, \delta, \iota)$ consists of the \emph{support} $sup: S \rightarrow \{0,1\}$ and the \emph{signature} $sig: E \rightarrow E_\tau$ where every edge $s \edge{e} s'$ of $A$ leads to an edge $sup(s) \ledge{sig(e)} sup(s')$ of type $\tau$.
If $P=q_0\edge{e_1}\dots \edge{e_n}q_n$ is a path in $A$, then $P^R=sup(q_0)\ledge{sig(e_1)}\dots\ledge{sig(e_n)}sup(q_n)$ is a path in $\tau$.
We say $P^R$ is the \emph{image} of $P$ (under $R$).
\end{definition}

\smallskip
Notice that $R$ is \emph{implicitly} completely defined by $sup(\iota)$ and $sig$:
Since $A$ is initialized, for every state $s\in S(A)$, there is a path $\iota\edge{e_1}\dots \edge{e_n}s_n$ such that $s=s_n$.
Thus, since $\delta_\tau$ is a (partial) function, we inductively obtain $sup(s_{i+1})$ by $sup(s_{i+1})=\delta_\tau(sup(s_i) ,sig(e_{i+1}))$ for all $i\in \{0,\dots, n-1\}$ and $s_0 = \iota$.
Consequently, we can compute $sup$ and thus $R$ purely from $sup(\iota)$ and $sig$ (and $A$), since when two paths lead to
the same state the corresponding supports are the same (otherwise $(sup,sig)$ does not define a region); this is illustrated by Example~\ref{ex:regions} and Figure~\ref{fig:image_of_path}.

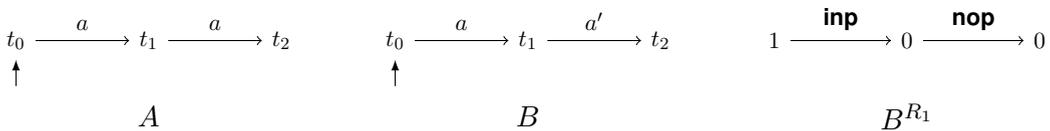
\begin{figure}[!h]
\begin{minipage}{\textwidth}
\begin{center}
\begin{tikzpicture}[new set = import nodes]
\begin{scope}[nodes={set=import nodes}]

\foreach \i in {0,...,2} {\node (\i) at (\i*1.75cm,0) {\nscale{$t_\i$}}; }
\node () at (1.75,-1) {$A$};
\draw[-latex](0,-0.6)to node[]{}(0);
\graph {
	(import nodes);
			0->["\escale{$a$}"]1->["\escale{$a$}"]2;
			};
\end{scope}
\begin{scope}[xshift=5cm,nodes={set=import nodes}]

\foreach \i in {0,...,2} {\node (\i) at (\i*1.75cm,0) {\nscale{$t_\i$}}; }
\draw[-latex](0,-0.6)to node[]{}(0);
\node () at (1.75,-1) {$B$};
\graph {
	(import nodes);
			0->["\escale{$a$}"]1->["\escale{$a'$}"]2;
			};
\end{scope}
\begin{scope}[xshift=10cm,nodes={set=import nodes}]

\foreach \i in {0} {\node (\i) at (\i*1.75cm,0) {\nscale{$1$}}; }
\foreach \i in {1,2} {\node (\i) at (\i*1.75cm,0) {\nscale{$0$}}; }
\node () at (1.75,-1) {$B^{R_1}$};
\graph {
	(import nodes);
			0->["\escale{$\inp$}"]1->["\escale{$\nop$}"]2;
			};
\end{scope}

\end{tikzpicture}
\end{center}\vspace*{-3mm}
\caption{%
Left: The TS $A$ with event set $E=\{a\}$.
Middle: The TS $B$ with event set  $E'=\{a,a'\}$.
Right: The image $B^{R_1}$ of the $\tau$-region $R_1=(sup_1,sig_1)$ of $B$, where $sup_1(t_0)=1$, $sup_1(t_1)=sup_1(t_2)=0$, $sig_1(a)=\inp$ and $sig_1(a')=\nop$, and $\tau=\{\nop,\inp,\swap\}$.
Later, we shall also represent a region by a color convention indicating the support of each state.
}\label{fig:image_of_path}
\end{minipage}
\end{figure}

\medskip
Every set $\R$ of $\tau$-regions of $A$ implies a particular \emph{synthesized} $\tau$-net, where the regions model places and the associated part of the flow-function:

\begin{definition}[Synthesized net]
Let $A=(S,E,\delta, \iota)$ be a TS, $\tau$ a Boolean type and $\R$ a set of $\tau$-regions of $A$.
The \emph{synthesized net} (indicated by $A$ and $\R$) is defined by $N^{\R}_A=(\R, E, f, M_0)$ where, for all $R=(sup, sig)\in \R$ and all $e\in E$, we have that  $f(R,e) = sig(e)$ and $M_0(R) = sup(\iota)$.
\end{definition}

It can be shown that there is always a (unique) simulation $\varphi$ between $A$ and the reachability graph  $A_{N^{\mathcal{R}}_A}$ of the synthesized net $N^{\mathcal{R}}_A$ where, for all $s\in S(A)$ and all $R\in \R$, we have that $sup(s)=M(R)$ for the marking $M$ of $N^{\mathcal{R}}_A$ that satisfies $\varphi(s)=M$.
However, to ensure that $\varphi$ is an embedding, we have to distinguish global states,
and to ensure that $\varphi$ is a language-simulation, we have to prevent the firings of transitions when their corresponding events are not present in TS.
This is stated as {\em separation atoms} and {\em properties}.

\begin{definition}[$\tau$-State Separation and Property]
A pair $(s, s')$ of distinct states of a $A$ defines a \emph{states separation atom} (SSA).
A $\tau$-region $R=(sup, sig)$ \emph{solves} $(s,s')$ if $sup(s)\not=sup(s')$; then this SSA is said $\tau$-solvable.
If every SSA of $A$ is $\tau$-solvable then $A$ has the \emph{$\tau$-states separation property} ($\tau$-SSP, for short).
\end{definition}

\begin{definition}[$\tau$-Event Separation and Property]
A pair $(e,s)$ of event $e\in E $ and state $s\in S$ where $e$ does not occur, that is $\neg s\edge{e}$, defines an \emph{event/state separation atom} (ESSA).
A $\tau$-region $R=(sup, sig)$ \emph{solves} $(e,s)$ if $sig(e)$ is not defined on $sup(s)$ in $\tau$, that is, $\neg sup(s)\ledge{sig(e)}$; then this ESSA is said $\tau$-solvable.
If every ESSA of $A$ is $\tau$-solvable then $A$ has the \emph{$\tau$-event state separation property} ($\tau$-ESSP, for short).
\end{definition}

\begin{definition}[$\tau$-Witness]
A set $\mathcal{R}$ of $\tau$-regions of $A$ is called a $\tau$-\emph{witness} of $A$'s $\tau$-SSP, respectively $\tau$-ESSP,
 if for each SSA, respectively ESSA, there is a $\tau$-region $R$ in $\mathcal{R}$ that solves it.
\end{definition}

The next lemma (\cite[p.~162]{txtcs/BadouelBD15}, Proposition 5.10) establishes the connection between the existence of $\tau$-witnesses and the existence of an implementing $\tau$-net $N$ for the witnessed property:
\begin{lemma}[\cite{txtcs/BadouelBD15}]\label{lem:admissible}
Let $A$ be a TS, $\tau$ a Boolean type and $N$ a $\tau$-net.
The following statements are true:
\begin{enumerate}
\item
 $A\hookrightarrow A_N$ if and only if there is a $\tau$-witness $\R$ of the $\tau$-SSP of $A$ and $N=N_A^{\R}$.
 \item
 $A\triangleright A_N$ if and only if there is a $\tau$-witness $\R$ of the $\tau$-ESSP of $A$ and $N=N_A^{\R}$.
 \item
 $A\cong A_N$ if and only if there is a $\tau$-witness $\R$ of both the $\tau$-SSP and the $\tau$-ESSP of $A$ and $N=N_A^{\R}$.
\end{enumerate}
\end{lemma}

\begin{example}\label{ex:regions}
Let $\tau$ be defined like in Figure~\ref{fig:type_net_rg} and $A,B,B^{R_1}$ like in Figure~\ref{fig:image_of_path}.
The TS $A$ has neither the $\tau$-SSP nor the $\tau$-ESSP, since the atoms $(t_0,t_2)$ and $(a, t_2)$ are not $\tau$-solvable.
The TS $B$ has both the $\tau$-SSP and $\tau$-ESSP.
The region $R_1=(sup_1,sig_1)$ that solves $(t_0,t_1)$, $(t_0,t_2)$, $(a,t_1)$ and $(a, t_2)$ is implicitly defined by $sup_1(t_0)=1$, $sig_1(a)=\inp$ and $sig_1(a')=\nop$.
We obtain $sup_1$ and thus $R_1$ explicitly by $sup_1(t_1)=\delta_{\tau}(1,\inp)=0$ and $sup_1(t_2)=\delta_{\tau}(0,\nop)=0$.
$B^{R_1}$ shows the image of $B$ under $R_1$.
The remaining (E)SSP atoms $(t_1,t_2)$, $(a',t_0)$ and $(a', t_1)$ are solved by the following $\tau$-region $R_2=(sup_2,sig_2)$ that is implicitly defined by $sup_2(t_0)=0$, $sig_2(a)=\swap$ and $sig_2(a')=\inp$.
The set $\mathcal{R}=\{R_1,R_2\}$ is a witness for the $\tau$-(E)SSP of $A$ and the net $N^\mathcal{R}_A$ is exactly the net $N$ that is depicted in Figure~\ref{fig:type_net_rg}.
$N$ is a realization of $A$, since a bijective simulation $\varphi$ between $A$ and $A_N$ is given by $\varphi(t_0)=(1,0)$, $\varphi(t_1)=(0,1)$ and $\varphi(t_2)=(0,0)$. \vspace*{-1mm}
\end{example}

\section{The complexity of  label-splitting}\label{sec:label_splitting}

For a Boolean type $\tau$, $\tau$-synthesis is the task to find, for a given TS $A$, a $\tau$-net $N$ that implements $A$.
Regardless of which of the implementations ($\tau$-embedding, $\tau$-language-simulation or $\tau$-realization) we are aiming at, a suitable $\tau$-net does not always exists.
In this case, \emph{label-splitting} might be a suitable technique to modify $A$ into a TS $B$ that is then implementable:

\begin{definition}[Label-splitting]
Let $A=(S,E,\delta, \iota)$ be a TS and $e_1,\dots, e_n\in E$ be pairwise distinct events.
The \emph{label-splitting} of the events $e_1,\dots, e_n$ into the events $e_1^1,\dots, e_1^{m_1}, \dots, e_n^1,\dots, e_n^{m_n}$  (pairwise distinct, and distinct from the other events in $E\setminus\{e_1,\dots,e_n\}$), where $m_j\geq2$ for all $j\in \{1,\dots, n\}$, yields the event set
$E'=(E\setminus \{e_1,\dots, e_n\})\cup\bigcup_{i=1}^{n}\{e_i^j \mid j\in \{1,\dots, m_j\}\}$.

\medskip
A TS $B=(S, E',\delta',\iota)$ is an $E'$-\emph{label-splitting} ($E'$-LS, for short) of $A$ if $\vert \delta \vert =\vert\delta' \vert$ and, for all $s,s'\in S$ and all $e\in E$, the following is true:
If  $\delta(s,e)=s'$ and $e\not\in \{e_1,\dots, e_n\}$, then $\delta'(s, e)=s'$;
if $\delta(s,e)=s'$ and $e=e_i$ for some $i\in \{1,\dots, n\}$, then there is exactly one $\ell\in \{1,\dots, m_i\}$ such that $\delta'(s, e_i^\ell)=s'$.
We say that $\mL=\{e_1,\dots, e_n\}$ is \emph{the set of events of $A$ that occur split in $B$}.
\end{definition}
Note that, in practice,  $e^1_j$ is usually chosen as the original event $e_j$ for some or all $j\in\{1,\ldots,n\}$.

\begin{example}\label{ex:label_splitting}
Let $A$ and $B$ be defined like in Figure~\ref{fig:image_of_path}.
The TS $B$ is an $E'$-label-splitting of $A$, where $E'=(E\setminus\{a\})\cup \{a,a'\}$.
\end{example}

To be as close as possible to the original behavior $A$, a corresponding $E'$-label-splitting $B$ of $A$ should change $A$ as little as possible, which means that the number of events of $B$ should be as small as possible.
This gives rise to consider label-splitting as a decision problem that, for a given TS $A$ and a natural number $\kappa$, asks whether there is an $E'$-label-splitting $B$ of $A$ that is implementable and uses at most $\kappa$ labels, i.e., $\vert E'\vert \leq \kappa$.

By Lemma~\ref{lem:admissible}, deciding the existence of an implementing net is equivalent to deciding if the input TS has the property that corresponds to the implementation.
Finally, this leads to the following three decision problems that are the main subject of this section:

\medskip\noindent
\fbox{\begin{minipage}[t][1.7\height][c]{0.88\textwidth}
\begin{decisionproblem}
  \problemtitle{\textsc{LS-$\tau$-Embedding}}
  \probleminput{a TS $A=(S,E,\delta, \iota)$, a natural number $\kappa$.}
  \problemquestion{Does there exist an $E'$-label-splitting $B$ of $A $ with $\vert E' \vert \leq \kappa$ that has the $\tau$-SSP?}
\end{decisionproblem}
\end{minipage}}

\bigbreak

\noindent
\fbox{\begin{minipage}[t][1.7\height][c]{0.88\textwidth}
\begin{decisionproblem}
  \problemtitle{\textsc{LS-$\tau$-Language-Simulation}}
  \probleminput{a TS $A=(S,E,\delta, \iota)$, a natural number $\kappa$.}
  \problemquestion{Does there exist an $E'$-label-splitting $B$ of $A $ with $\vert E' \vert \leq \kappa$ that has the {$\tau$-ESSP?}}
\end{decisionproblem}
\end{minipage}}

\bigbreak

\noindent
\fbox{\begin{minipage}[t][1.7\height][c]{0.88\textwidth}
\begin{decisionproblem}
  \problemtitle{\textsc{LS-$\tau$-Realisation}}
  \probleminput{a TS $A=(S,E,\delta, \iota)$, a natural number $\kappa$.}
  \problemquestion{Does there exist an $E'$-label-splitting $B$ of $A $ with $\vert E' \vert \leq \kappa$ that has both the $\tau$-SSP and the $\tau$-ESSP?}
\end{decisionproblem}
\end{minipage}}

\vspace{3mm}

The following theorem presents the main result of this section:\smallskip

\begin{theorem}\label{the:label_splitting}
If $\tau=\{\nop,\swap\}\cup\omega$ with $\omega\subseteq \{\inp,\out,\used,\free\}$, then
\begin{enumerate}
\item\label{the:label_splitting_embedding}
\textsc{LS-$\tau$-Embedding} is NP-complete,
\item\label{the:label_splitting_ls_and_real}
\textsc{LS-$\tau$-Language-Simulation}, and \textsc{LS-$\tau$-Realisation} are NP-complete if $\omega\not=\emptyset$, otherwise they are in P.
\end{enumerate}
\end{theorem}

First of all, we argue for the polynomial part, where we need the $\tau$-ESSP:
If $\tau=\{\nop,\swap\}$\footnote{This remains true if $\tau\subseteq\{\nop,\swap,\set,\res\}$, but $\set/\res$ are not used in this paper.}, then a TS $A=(S,E,\delta,\iota)$ has the $\tau$-ESSP if and only if every event occurs at every state.
Indeed, if every event occurs at every state there is no ESSA that needs to be solved.
On the contrary, if some event $e$ is missing from some state $s$, since all the allowed functions are defined on both $0$ and $1$,
 it is not possible to solve the ESSA $(e,s)$ of $A$.
This implies that an $E'$-label-splitting $B$ of $A$ that reflects an actual splitting, meaning that $\vert E'\vert > \vert E\vert$, can never yield a TS that has the $\tau$-ESSP, since it would produce at least one unsolvable ESSA:
If $\vert E'\vert > \vert E\vert$, then there is an event $e_i$ in $A$ for which there are edges $s\Edge{e_i^j}s'$, and $t\Edge{e_i^\ell}t'$ in $B$ with $s\not=t$, and $e_i^j\not=e_i^\ell$ (since $\delta_A$, and $\delta_B$ are functions, and every event has an occurrence by Definition~\ref{def:ts}), which were previously both labeled by $e_i$.
This implies $\neg s\Edge{e_i^\ell}$ (by $\delta_A$ being a function),
and thus $(s, e_i^\ell)$ is not solvable.
The same argument shows that, for the realization problem, if $A$ has the ESSP (meaning that every event occurs at every state as explained before) but has an unsolvable SSA, then any label-splitting introduced to solve this atom would destroy the ESSP.
Hence, for this Boolean type, the decision problems are trivial: either $A$ is already implementable (which can be can be checked in polynomial time~\cite{stacs/Schmitt96,tamc/TredupR19}) or it has to be rejected.
Thus, for the proof of Theorem~\ref{the:label_splitting}(\ref{the:label_splitting_ls_and_real}), it remains to consider the cases when $\omega\not=\emptyset$, hence the NP-completeness results.

The decision problems \textsc{LS-$\tau$-Embedding}, \textsc{LS-$\tau$-Language-Simulation} as well as \textsc{LS-$\tau$-Realisation} are in NP:
If a sought $E'$-label-splitting $B=(S,E', \delta',\iota)$ of a TS $A=(S,E, \delta,\iota )$ exists, then a Turing-machine $M$ can compute $B$ in a non-deterministic computation in time polynomial in the size $\vert \delta\vert$ of $A$, since $\vert \delta\vert =\vert \delta'\vert$.
After that, $M$ verifies in time polynomial in the size $\vert \delta'\vert$ of $B$ (and thus of $A$) that it allows the sought implementation~\cite{stacs/Schmitt96,tamc/TredupR19}.\medskip

Recall that a \emph{non-directed graph} $G=(\U,M)$ consists of a (finite) set $\U$ of \emph{vertices}, and a set $M$ of (non-directed) \emph{edges} over $\U$, that is, for all $\mathfrak{e}\in M$, we have that $\mathfrak{e}\subseteq \U$, and $\vert \mathfrak{e}\vert=2$.

\medskip
Our NP-completeness proofs are based on reductions from the following classical variant of the vertex cover (VC) problem~\cite[p.~190]{fm/GareyJ79}:

\noindent\medskip
\fbox{\begin{minipage}[t][1.7\height][c]{0.88\textwidth}
\begin{decisionproblem}
  \problemtitle{\textsc{3-Bounded Vertex Cover}~(\textsc{3BVC})}
  \probleminput{a non-directed  Graph $G=(\U, M)$ such that every vertex $v\in \U$ is a member of at most three distinct edges, and a natural number $\lambda\in \mathbb{N}$.}
  \problemquestion{Does there exist a $\lambda$-vertex cover ($\lambda$-VC, for short)
  of $G$, that is a subset $\Ss\subseteq \U$ with $\vert \Ss\vert \leq \lambda$ and $\Ss\cap \mathfrak{e}\not=\emptyset$ for all $\mathfrak{e}\in M$?}
\end{decisionproblem}
\end{minipage}}

\begin{example}[\textsc{3BVC}]\label{ex:vc}
The instance $(G,2)$, illustrated in Figure~\ref{G.fig}, where $G=(\U,M)$ such that $\U=\{v_0,v_1,v_2,v_3\}$ and $M=\{ M_0,\dots, M_4\}$, where $M_0=\{v_0,v_1\}, M_1=\{v_0,v_2\}, M_2=\{v_0,v_3\}, M_3=\{v_1,v_2\}$, and $M_4=\{v_2,v_3\} $, is a yes-instance of \textsc{3BVC}, since $\Ss=\{v_0,v_2\}$ is a 2-VC of $G$.
\end{example}

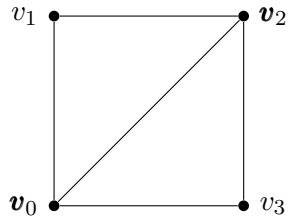
\begin{figure}[H]
 \tikzstyle{every label}=[black]
\centering
\begin{tikzpicture}[yscale=0.5,xscale=0.5]
\node[circle,fill=black!100,inner sep=0.05cm](0)at(0,0)[label=left:$\pmb v_0$]{};
\node[circle,fill=black!100,inner sep=0.05cm](1)at(0,5)[label=left:$v_1$]{};
\node[circle,fill=black!100,inner sep=0.05cm](2)at(5,5)[label=right:$\pmb v_2$]{};
\node[circle,fill=black!100,inner sep=0.05cm](3)at(5,0)[label=right:$v_3$]{};
\draw[-](0)to node[auto,left]{}(1);
\draw[-](0)to node[auto,above,pos=0.3]{}(2);
\draw[-](0)to node[auto,right,pos=0.3]{}(3);
\draw[-](1)to node[auto,right]{}(2);
\draw[-](2)to node[auto,right]{}(3);
\end{tikzpicture}
\caption{A running graph example $G$; a 2-VC is $\{{\pmb v_0},{\pmb v_2}\}$ (in bold).}
\label{G.fig}
\end{figure}

In the remainder of this paper, let $(G,\lambda)$ be an input of \textsc{3BVC}, where $G=(\U,M)$ is a graph with $n$ vertices $\U=\{v_0,\dots, v_{n-1}\}$ and $m$ edges $M=\{M_0,\dots, M_{m-1}\}$ such that $M_i=\{v_{i_0}, v_{i_1}\}$ and $i_0 < i_1$ for all $i\in \{0,\dots, m-1\}$.

\medskip
For the proof of Theorem~\ref{the:label_splitting}, we polynomially reduce $(G,\lambda)$ to a pair $(A_G,\kappa)$ of TS
$A_G=(S,E,\delta,\bot_0)$ and natural number $\kappa$ such that the following conditions are satisfied:
\begin{enumerate}
\itemsep=0.95pt
\item\label{con:implementation_implies_vc}
If there is an $E'$-label-splitting of $A_G$ that satisfies $\vert E'\vert \leq \kappa$ and has the $\tau$-SSP or the $\tau$-ESSP, then $G$ has a $\lambda$-VC.
\item\label{con:vc_implies_implementation}
If $G$ has a $\lambda$-VC, then there is an $E'$-label-splitting of $A_G$ that satisfies $\vert E'\vert \leq \kappa$ and has both the $\tau$-SSP and the $\tau$-ESSP.
\end{enumerate}
Obviously, a polynomial-time reduction that satisfies Condition~\ref{con:implementation_implies_vc} and Condition~\ref{con:vc_implies_implementation} ensures that $G$ allows a $\lambda$-VC if and only if $A_G$ allows an $E'$-label-splitting that satisfies $\vert E'\vert \leq \kappa$ and has the $\tau$-SSP, the $\tau$-ESSP or both, according to which property is sought.
Hence, it proves Theorem~\ref{the:label_splitting}.

\subsection{The proof of Theorem~\ref{the:label_splitting}(\ref{the:label_splitting_embedding})}\label{sec:label_splitting_ssp}%

In the remainder of this section we assume that $\tau=\{\nop,\swap\}\cup\omega$, with an arbitrary but fixed subset $\omega\subseteq\{\inp,\out,\used,\free\}$.

\medskip
For a start, we define $\kappa=n+2m-1+\lambda$, where $n+2m-1$ is the number of events of the aforementioned TS $A_G$.
Hence, $\lambda$ is the maximum number of events of $A_G$ that could potentially be split in an $E'$-label-splitting $B_G$ of $A_G$.

\medskip
For every $i\in \{0,\dots, m-1\}$, the TS $A_G$ has the following directed path $T_i$ that uses the vertices $v_{i_0}$ and $v_{i_1}$ of the edge $M_i$ as events:
\begin{center}
\begin{tikzpicture}[new set = import nodes, scale =0.8]
\begin{scope}[nodes={set=import nodes}]%
		\node (T) at (-0.9,0) {$T_i=$};
		\foreach \i in {0,...,4} { \coordinate (\i) at (\i*2cm, 0) ;}
		\foreach \i in {0,...,4} { \node (\i) at (\i) {\nscale{$t_{i,\i}$}};}
\graph {(import nodes);
			0 ->["\escale{$v_{i_0}$}"]1->["\escale{$v_{i_1}$}"]2->["\escale{$v_{i_0}$}"]3->["\escale{$v_{i_1}$}"]4;
		};
\end{scope}
\end{tikzpicture}
\end{center}
Finally, for all $i\in \{0,\dots, m-1\}$, we apply the edge $\bot_i\edge{w_i}t_{i,0}$ and, if $i < m-1$, then also the edge $\bot_{i}\edge{\ominus_{i+1}}\bot_{i+1}$ to connect the paths $T_0,\dots, T_{m-1}$ into the TS $A_G$, cf.~Figure~\ref{fig:ex_A_and_B}.
Let $\bot=\{\bot_0,\dots, \bot_{m-1}\}$ and $W=\{w_0,\dots, w_{m-1}\}$ and $\ominus=\{\ominus_1,\dots, \ominus_{m-1}\}$.
The TS $A_G$ has exactly $\vert V \cup W\cup \ominus\vert= n+2m-1$ events.

\begin{figure}[htb]
\begin{center}
\begin{tikzpicture}[new set = import nodes]
\begin{scope}[nodes={set=import nodes}]
		\coordinate (bot1) at (0,0);
		\foreach \i in {0,...,4} {\coordinate (\i) at (\i*1.8cm+1.8cm,0);}
		\node (bot1) at (bot1) {\nscale{$\bot_0$}};
		\foreach \i in {0,...,4} {\node (\i) at (\i) {\nscale{$t_{0,\i}$}};}
		\draw[-latex](-0.6,0)to node[]{}(bot1);
		\graph {
	(import nodes);
			bot1->["\escale{$w_0$}"]0->["\escale{$v_0$}"]1->["\escale{$v_1$}"]2->["\escale{$v_0$}"]3->["\escale{$v_1$}"]4;

			};
\end{scope}
\begin{scope}[yshift=-1cm, nodes={set=import nodes}]
		\coordinate (bot2) at (0,0);
		\foreach \i in {0,...,4} {\coordinate (\i) at (\i*1.8cm+1.8cm,0);}
		\node (bot2) at (bot2) {\nscale{$\bot_1$}};
		\foreach \i in {0,...,4} {\node (\i) at (\i) {\nscale{$t_{1,\i}$}};}		
		\graph {
	(import nodes);
			bot2->["\escale{$w_1$}"]0->["\escale{$v_0$}"]1->["\escale{$v_2$}"]2->["\escale{$v_0$}"]3->["\escale{$v_2$}"]4;
			};
\end{scope}
\begin{scope}[yshift=-2cm, nodes={set=import nodes}]
		\coordinate (bot3) at (0,0);
		\foreach \i in {0,...,4} {\coordinate (\i) at (\i*1.8cm+1.8cm,0);}
		\node (bot3) at (bot3) {\nscale{$\bot_2$}};
		\foreach \i in {0,...,4} {\node (\i) at (\i) {\nscale{$t_{2,\i}$}};}				
		\graph {
	(import nodes);
			bot3->["\escale{$w_2$}"]0->["\escale{$v_0$}"]1->["\escale{$v_3$}"]2->["\escale{$v_0$}"]3->["\escale{$v_3$}"]4;
			};
\end{scope}
\begin{scope}[yshift=-3cm, nodes={set=import nodes}]%
		\coordinate (bot4) at (0,0);
		\foreach \i in {0,...,4} {\coordinate (\i) at (\i*1.8cm+1.8cm,0);}
		\node (bot4) at (bot4) {\nscale{$\bot_3$}};
		\foreach \i in {0,...,4} {\node (\i) at (\i) {\nscale{$t_{3,\i}$}};}					
		\graph {
	(import nodes);
			bot4->["\escale{$w_3$}"]0->["\escale{$v_1$}"]1->["\escale{$v_2$}"]2->["\escale{$v_1$}"]3->["\escale{$v_2$}"]4;
			
			};
\end{scope}
\begin{scope}[yshift=-4cm,nodes={set=import nodes}]%
		\coordinate (bot5) at (0,0);
		\node (bot5) at (0,0) {\nscale{$\bot_4$}};
		\foreach \i in {0,...,4} { \coordinate (\i) at (1.8cm+ \i*1.8cm,0) ;}
		\foreach \i in {0,...,4} { \node (\i) at (\i) {\nscale{$t_{4,\i}$}};}
		\graph {
	(import nodes);
			bot5->["\escale{$w_4$}"]0->["\escale{$v_2$}"]1->["\escale{$v_3$}"]2->["\escale{$v_2$}"]3->["\escale{$v_3$}"]4;
			};
\end{scope}

\path (bot1) edge [->] node[left] {\nscale{$\ominus_1$} } (bot2);
\path (bot2) edge [->] node[left] {\nscale{$\ominus_2$} } (bot3);
\path (bot3) edge [->] node[left] {\nscale{$\ominus_3$} } (bot4);
\path (bot4) edge [->] node[left] {\nscale{$\ominus_4$} } (bot5);
\end{tikzpicture}
\end{center}

\begin{center}
\begin{tikzpicture}[new set = import nodes]
\begin{scope}[nodes={set=import nodes}]
		\coordinate (bot1) at (0,0);
		\foreach \i in {0,...,4} {\coordinate (\i) at (\i*1.8cm+1.8cm,0);}
		\node (bot1) at (bot1) {\nscale{$\bot_0$}};
		\foreach \i in {0,...,4} {\node (\i) at (\i) {\nscale{$t_{0,\i}$}};}
		\draw[-latex](-0.6,0)to node[]{}(bot1);	
		\graph {
	(import nodes);
			bot1->["\escale{$w_0$}"]0->["\escale{$v_0'$}"]1->["\escale{$v_1$}"]2->["\escale{$v_0$}"]3->["\escale{$v_1$}"]4;
			
			};
\end{scope}
\begin{scope}[yshift=-1cm, nodes={set=import nodes}]
		\coordinate (bot2) at (0,0);
		\foreach \i in {0,...,4} {\coordinate (\i) at (\i*1.8cm+1.8cm,0);}
		\node (bot2) at (bot2) {\nscale{$\bot_1$}};
		\foreach \i in {0,...,4} {\node (\i) at (\i) {\nscale{$t_{1,\i}$}};}			
		\graph {
	(import nodes);
			bot2->["\escale{$w_1$}"]0->["\escale{$v_0$}"]1->["\escale{$v_2$}"]2->["\escale{$v_0'$}"]3->["\escale{$v_2'$}"]4;
			};
\end{scope}
\begin{scope}[yshift=-2cm, nodes={set=import nodes}]
		\coordinate (bot3) at (0,0);
		\foreach \i in {0,...,4} {\coordinate (\i) at (\i*1.8cm+1.8cm,0);}
		\node (bot3) at (bot3) {\nscale{$\bot_2$}};
		\foreach \i in {0,...,4} {\node (\i) at (\i) {\nscale{$t_{2,\i}$}};}				
		\graph {
	(import nodes);
			bot3->["\escale{$w_2$}"]0->["\escale{$v_0'$}"]1->["\escale{$v_3$}"]2->["\escale{$v_0$}"]3->["\escale{$v_3$}"]4;
			};
\end{scope}
\begin{scope}[yshift=-3cm, nodes={set=import nodes}]%
		\coordinate (bot4) at (0,0);
		\foreach \i in {0,...,4} {\coordinate (\i) at (\i*1.8cm+1.8cm,0);}
		\node (bot4) at (bot4) {\nscale{$\bot_3$}};
		\foreach \i in {0,...,4} {\node (\i) at (\i) {\nscale{$t_{3,\i}$}};}					
		\graph {
	(import nodes);
			bot4->["\escale{$w_3$}"]0->["\escale{$v_1$}"]1->["\escale{$v_2$}"]2->["\escale{$v_1$}"]3->["\escale{$v_2'$}"]4;
			
			};
\end{scope}
\begin{scope}[yshift=-4cm,nodes={set=import nodes}]%
		\coordinate (bot5) at (0,0);
		\node (bot5) at (0,0) {\nscale{$\bot_4$}};
		\foreach \i in {0,...,4} { \coordinate (\i) at (1.8cm+ \i*1.8cm,0) ;}
		\foreach \i in {0,...,4} { \node (\i) at (\i) {\nscale{$t_{4,\i}$}};}
		\graph {
	(import nodes);
			bot5->["\escale{$w_4$}"]0->["\escale{$v_2'$}"]1->["\escale{$v_3$}"]2->["\escale{$v_2$}"]3->["\escale{$v_3$}"]4;
			};
\end{scope}

\path (bot1) edge [->] node[left] {\nscale{$\ominus_1$} } (bot2);
\path (bot2) edge [->] node[left] {\nscale{$\ominus_2$} } (bot3);
\path (bot3) edge [->] node[left] {\nscale{$\ominus_3$} } (bot4);
\path (bot4) edge [->] node[left] {\nscale{$\ominus_4$} } (bot5);
\end{tikzpicture}
\end{center}\vspace*{-2mm}
\caption{ \label{fig:ex_A_and_B}
The transition systems  $A_G$ (top) and  $B_G$ (bottom) that originate from Example~\ref{ex:vc}.
They will serve for illustrating some region constructions in the proofs below.}
\end{figure}
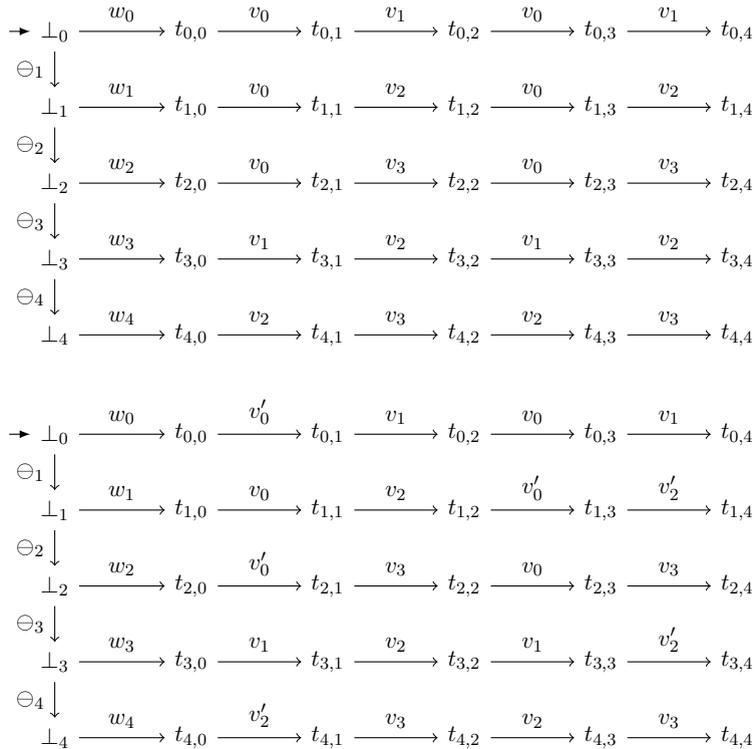

In order to prove Theorem~\ref{the:label_splitting}(\ref{the:label_splitting_embedding}) we shall show that $A_G$ allows a label-splitting $B_G$ restricted by $\kappa$, and having the $\tau$-SSP if and only if there is a $\lambda$-VC  for $G$.

The following lemma implies that if a TS has any of the introduced paths, then it does not have the $\tau$-SSP:
\begin{lemma}\label{lem:abab}
Let $A$ be a TS that has the path $P_0=s_0\edge{a}s_1\edge{b}s_2\edge{a}s_3\edge{b}s_4$.
If $R=(sup, sig)$ is a $\tau$-region of $A$, then $sup(s_0)=sup(s_4)$.
\end{lemma}
\begin{proof}
Let $R=(sup, sig)$ be an arbitrary but fixed region of $A$.
If $sup(s_0)\not=sup(s_4)$, then the image $P^R$ is a path from $0$ to $1$ or from $1$ to $0$ in $\tau$.
This implies that the number of state changes between $0$ and $1$ on $P^R$ must be odd.
Since $sig(a), sig(b)\in \{\nop,\inp,\out,\swap,\used,\free\}$ and both $a$ and $b$ occur twice, i.e. an even number of times, this is impossible.
\end{proof}

Note that this lemma would not be true if {\set} and/or {\res} would be allowed in $\tau$:
for instance, we have $0\edge{\set}1\edge{\nop}1\edge{\set}1\edge{\nop}1$.

\smallskip
Hence, if a TS has the path $T_i$ for some $i\in \{0,\dots, m-1\}$, then the SSP atom $(t_{i,0},t_{i,4})$ is not $\tau$-solvable by Lemma~\ref{lem:abab}.
The following lemma states that this implies a $\lambda$-VC of $G$ if a sought $E'$-label-splitting $B_G$ of $A_G$ exists:
\begin{lemma}\label{lem:ssp_implies_vc}
If there is an $E'$-label-splitting $B_G$ of $A_G$ such that $\vert E' \vert \leq \kappa$ that has the $\tau$-SSP, then $G$ has a $\lambda$-VC.
\end{lemma}
\begin{proof}
Let $i\in \{0,\dots, m-1\}$ be arbitrary but fixed, and let $\mL$ be the set of events of $A_G$ that occur split in $B_G$.
By Lemma~\ref{lem:abab}, the SSP atom $\alpha_i=(t_{i,0},t_{i,4})$ is not $\tau$-solvable by regions of $A_G$.
However, since $B_G$ has the $\tau$-SSP, the atom $\alpha_i$ is $\tau$-solvable in $B_G$.
This implies $\{v_{i_0}, v_{i_1}\}\cap \mL \not=\emptyset$.
Since $i$ was arbitrary, this is simultaneously true for all paths $T_0,\dots, T_{m-1}$ and thus the set $\Ss=\U\cap \mL$ intersects with every edge of $G$.
Moreover, by $\vert E'\vert \leq \kappa=n+2(m-1)+\lambda$, we have $\vert \Ss\vert \leq \vert \mL\vert \leq \lambda$.
Hence, $\Ss$ defines a $\lambda$-VC of $G$.
\end{proof}
Conversely, let $\Ss=\{v_{j_0},\dots, v_{j_{\lambda-1}}\}\subseteq \U$ be a $\lambda$-VC of $G$.
(Notice that a $\lambda$-VC with less than $\lambda$ states implies one with exactly $\lambda$ states as longs as $\lambda \leq \vert \U\vert$, which can be reasonably assumed.)
In the remainder of this section, we argue that there is a sought $E'$-label-splitting $B_G$ of $A_G$.
For every $i\in \{0,\dots, \lambda-1\}$, we split the event $v_{j_i}$ into the two events $v_{j_i}$  and $v_{j_i}'$.
This yields $E'=(E\setminus \Ss)\cup\bigcup_{i=0}^{\lambda-1}\{v_{j_i}, v_{j_i}'\}$.
To define the aforementioned $E'$-label-splitting $B_G=(S, E',\delta', \bot_0)$ of $A_G$, it suffices to define $\delta'$ on the states of $T_0,\dots, T_{m-1}$.
In particular, for all $i\in \{0,\dots, m-1\}$, $\delta'$ restricted to $S(T_i)$ and $E(T_i)$ yields the path $T_i'$ as follows:
\begin{itemize}
\itemsep=0.95pt
\item
if $v_{i_0}\in \Ss$ and $v_{i_1}\not\in \Ss$, then $T_i'=t_{i,0}\lEdge{v_{i_0}'}t_{i,1}\lEdge{v_{i_1}}t_{i,2},\lEdge{v_{i_0}}t_{i,3}\lEdge{v_{i_1}} t_{i,4}$;
\item
if $v_{i_0}, v_{i_1}\in \Ss$, then $T_i'=t_{i,0}\lEdge{v_{i_0}}t_{i,1}\lEdge{v_{i_1}}t_{i,2},\lEdge{v_{i_0}'}t_{i,3}\lEdge{v_{i_1}'} t_{i,4}$;
\item
if $v_{i_0}\not\in \Ss$ and $v_{i_1}\in \Ss$, then $T_i'=t_{i,0}\lEdge{v_{i_0}}t_{i,1}\lEdge{v_{i_1}}t_{i,2},\lEdge{v_{i_0}}t_{i,3}\lEdge{v_{i_1}'} t_{i,4}$.
\end{itemize}
\eject
The following lemma essentially states that if $sup: S(B_G)\rightarrow\{0,1\}$ and $sig: E(B_G)\rightarrow \tau$ are mappings that define regions when restricted to $T_0',\dots, T_{m-1}'$, then they can be extended suitably to a region of $B_G$:
\begin{lemma}\label{lem:extend_region_1}
If $sup: S(B_G)\setminus \bot \rightarrow \{0,1\}$ and $sig: E'\setminus (W\cup\ominus) \rightarrow \tau$ are mappings such that $s\edge{e}s'\in B_G$ and $e\not\in W\cup\ominus$ imply $sup(s)\ledge{sig(e)}sup(s')\in A_\tau$,
then there is a $\tau$-region $R=(sup', sig')$ of $B_G$ that preserves $sup$ and $sig$ as follows:
\begin{enumerate}
\item
For all $s\in S(B_G)$, if $s\not\in \bot$ then $sup'(s)=sup(s)$, otherwise $sup'(s)=0$.
\item
For all $e\in E'$, if $e\not\in W\cup\ominus$ then $sig'(e)=sig(e)$;
if $e\in \ominus$, then $sig(e)=\nop$;
if $i \in \{0,\dots, m-1\}$ and $e=w_i$, if $sup(t_{i,0})=0$ then $sig(e)=\nop$, otherwise $sig(e)=\swap$.
\end{enumerate}
\end{lemma}
\begin{proof}
We argue that $s\edge{e}s'\in B_G$ implies $sup'(s)\ledge{sig'(e)}sup'(s')\in \tau$.
If $e\in W\cup\ominus$, then this is easy to see, since $sup'(\bot_i)=0$ for all $i\in \{0,\dots, m-1\}$.
For $e\not\in W\cup\ominus$, the claim follows by the assumptions about $sup$ and $sig$.
\end{proof}

By the next lemma, a $\tau$-region of $T_i'$, where $T_i'$ is considered as a TS, whose signature only uses $\nop$ and $\swap$ is always extendable to a region of $B_G$:
\begin{lemma}\label{lem:extend_region_2}
Let $i\in \{0,\dots, m-1\}$.
Let $sup: S(T_i') \rightarrow \{0,1\}$ and $sig: E(T_i') \rightarrow \{\nop,\swap\}$ be mappings such that $s\edge{e}s'\in T_i'$ implies $sup(s)\ledge{sig(e)}sup(s')\in \tau$.
There is a $\tau$-region $R=(sup', sig')$ of $B_G$ such that $sup'(s)=sup(s)$ and $sig'(e)=sig(e)$ for all $s\in S(T_i')$ and $e\in E(T_i')$.
\end{lemma}
\begin{proof}
By Lemma~\ref{lem:extend_region_1}, it suffices to argue that $sup$ and $sig$ are consistently extendable to
 $T_0',\dots T_{i-1}',$ $T_{i+1}',\dots, T_{m-1}'$.
Let $j\in \{0,\dots, m-1\}\setminus\{i\}$, be arbitrary but fixed,  and let \\ $T_j'=t_{j,0}\edge{e_{j,1}}t_{j,1}\edge{e_{j,2}}t_{j,2}\edge{e_{j,3}}t_{j,3}\edge{e_{j,4}}t_{j,4}$, where $e_{j,1},\dots,e_{j,4}\in E(T_j')$ in accordance to the definition of $B_G$.
We obtain $R=(sup',sig')$ as follows.
For all $e\in E(B_G)\setminus (W\cup\ominus)$, if $e\in E(T_i')$, then $sig'(e)=sig(e)$ and otherwise $sig(e)=\nop$;
for all $s\in S(T_i')$, we define $sup'(s)=sup(s)$;
for all $j\in \{0,\dots, m-1\}\setminus\{i\}$, we define $sup(t_{j,0})=0$ and inductively $sup(t_{j,\ell})=\delta_G(sup(t_{j,\ell-1}), e_{j,\ell})$ for all $\ell\in \{1,\dots,4\}$.
Since $sig$ maps to $\{\nop,\swap\}$, so does $sig'$.
Thus, if $s\edge{e}s'\in T_j'$, then $sup'(s)\ledge{sig'(e)}sup'(s')\in A_\tau$.
Since j was arbitrary, this proves the lemma.
\end{proof}

\begin{lemma}\label{lem:vc_implies_ssp}
The TS $B_G$ has the $\tau$-SSP.
\end{lemma}
\begin{proof}
It is easy to see that $(\bot_i,s)$ is $\tau$-solvable for all $i\in \{0,\dots,m-1\}$ and all $s\in S(B_G)\setminus \{\bot_i\}$:
one may choose $sup(\bot_i)=1$, $sup(s)=0$ if $s\neq\bot_i$, $sig(e)=\swap$ if $\edge{e}\bot_i$ or $\bot_i\edge{e}$, $sig(e)=\nop$ otherwise.
\eject

\noindent  Similarly, the atom $(s,s')$ where $s\in S(T_i')$ and $s'\in S(B_G)\setminus S(T_i')$ is $\tau$-solvable for all $i\in \{0,\dots, m-1\}$:
one may choose $sup(s)=1$ if $s\in S(T_i)$ and $0$ otherwise, $sig(e)=\swap$ if $e=w_i$ and $\nop$ otherwise.

Thus, it remains to argue that an atom $(s, s')$ is also solvable if $s\not=s'\in S(T_i')$, for all $i\in \{0,\dots, m-1\}$.
By Lemma~\ref{lem:extend_region_2}, it suffices to present corresponding regions for $T_i'$.

Let $i\in \{0,\dots, m-1\}$ be arbitrary but fixed, and, for a start, let's consider the case where $v_{i_0}\in M$ and $v_{i_1}\not\in M$.
That is, $T_i'=t_{i,0}\Edge{v_{i_0}'}t_{i,1}\Edge{v_{i_1}}t_{i,2}\Edge{v_{i_0}}t_{i,2}\Edge{v_{i_1}}t_{i,4}$.
For all $\ell\in \{1,2,3\}$, let $R_\ell=(sup_\ell, sig_\ell)$ be a pair of mappings $sup_\ell:S(T_i')\rightarrow \{0,1\}$, $sig_\ell:E(T_i')\rightarrow \{\nop,\swap\}$, (implicitly) defined by $sup_1(t_{i,0})=0$, $sig_1(v_{i_0}')=sig_1(v_{i_0})=\nop$ and $sig_1(v_{i_1})=\swap$;
and $sup_2(t_{i,0})=0$, $sig_2(v_{i_0}')=sig_2(v_{i_1})=\nop$ and $sig_2(v_{i_0})=\swap$;
 and $sup_3(t_{i,0})=0$, $sig_3(v_{i_0})=sig_3(v_{i_1})=\nop$ and $sig_3(v_{i_0}')=\swap$.
The images of $T_i'$ under $R_1,R_2$ and $R_3$ are as follows:
\begin{align*}
{T_i'}^{R_1}&=0\ledge{\nop}0\ledge{\swap}1\ledge{\nop}1\ledge{\swap}0
\textcolor{white}{ and } {T_i'}^{R_2}=0\ledge{\nop}0\ledge{\nop}0\ledge{\swap}1\ledge{\nop}1\\
{T_i'}^{R_3}&=0\ledge{\swap}1\ledge{\nop}1\ledge{\nop}1\ledge{\nop}1
\end{align*}
By Lemma~\ref{lem:extend_region_2}, $R_\ell$ can be extended to a $\tau$-region of $B_G$ that preserves $sup$ for all $\ell\in \{1,2,3 \}$.
Moreover, obviously, for every SSP atom $(s,s')$ of $T_i'$, there is an $\ell\in \{1,2,3\}$ such that $sup_\ell(s)\not= sup_\ell(s')$.
Thus, $(s,s')$ is $\tau$-solvable in $B_G$.

\medskip
The arguments for the case $v_{i_0}\not\in M$ and $v_{i_1}\in M$ is similar (the situation is symmetrical);
the case  $v_{i_0},v_{i_1}\in M$ is simpler since no two events are the same.

By the arbitrariness of $i$, this proves the lemma.
\end{proof}

\subsection{The proof of Theorem~\ref{the:label_splitting}(\ref{the:label_splitting_ls_and_real}) when $\tau\cap\{\inp,\out\}\not=\emptyset$}\label{sec:label_splitting_essp}%

Let $\tau=\{\nop,\swap\}\cup\omega$ be a type of nets such that $\omega\subseteq\{\inp,\out,\used,\free\}$ and $\omega\cap\{\inp,\out\}\not=\emptyset$, and let $A_G$ and $\kappa$ be defined as in Section~\ref{sec:label_splitting_ssp}.

\begin{lemma}\label{lem:essp_implies_vc}
If there is an $E'$-label-splitting $B_G$ of $A_G$ such that $\vert E'\vert\leq \kappa$ that has the $\tau$-ESSP, then $G$ has a $\lambda$-VC.
\end{lemma}
\begin{proof}
Let $i\in \{0,\dots, m-1\}$ be arbitrary but fixed and $\mL$ be the set of events of $A_G$ that occur split in $B_G$.
If $R=(sup, sig)$ is a $\tau$-region of $A_G$, then $sup(t_{i,0})=sup(t_{i,4})$ by Lemma~\ref{lem:abab}.
Thus, $\alpha_i=(v_{i_0}, t_{i,4})$ is not $\tau$-solvable, since $sup(t_{i,0})\lEdge{sig(v_{i_0})}$ implies $sup(t_{i,4})\lEdge{sig(v_{i_0})}$.
On the other hand, $B_G$ has the $\tau$-ESSP, implying the $\tau$-solvability of $\alpha_i$.
This implies $\{v_{i_0}, v_{i_1}\}\cap\mL\not=\emptyset$.
Since $i$ was arbitrary, this is  true for all $T_0,\dots, T_{m-1}$.
Thus, just like for Lemma~\ref{lem:ssp_implies_vc}, we get that $\Ss=\mL\cap \U$ defines a $\lambda$-VC of $G$.
This proves the lemma.
\end{proof}

Conversely, let $\Ss$ be a $\lambda$-VC of $G$, and $B_G$ be the $E'$-label-splitting of $A_G$ as defined in Section~\ref{sec:label_splitting_ssp}.
By the following lemma, $B_G$ has an exact net realization:

\begin{lemma}\label{lem:vc_implies_essp}
The TS $B_G$ has the $\tau$-SSP, and the $\tau$-ESSP.
\end{lemma}
\begin{proof}
By Lemma~\ref{lem:vc_implies_ssp}, the TS $B_G$ has the $\tau$-SSP.
It remains to argue for the $\tau$-ESSP.
Without loss of generality, we assume that $\inp\in\tau$ and present $\tau$-regions $R=(sup, sig)$ that only use $\nop,\inp$ and $\swap$.
Indeed, if $\inp\not\in \tau$, then $\out\in \tau$ and one gets corresponding (complement) regions $R'=(sup', sig')$ simply by $sup'(s)=1-sup(s)$, $sig'(e)=sig(e)$ if $sig(e)\in \{\nop,\swap\}$, and $sig'(e)=\out$ if $sig(e)=\inp$ for all $s\in S(B_G)$ and all $e\in E'$.

\medskip
The general idea to solve an ESSA $(e,s)$ is to choose $sup(s)=0$, $sig(e)=\inp$,
and $sig(e')\in\{\nop,\swap\}$ as needed, and $sup(s')$ accordingly for the other events $e'$, and states $s'$ of $A_G$.
For instance, if $e\in W \cup\ominus$, one may choose $sup(s')=1$ if $s'\edge{e}$ and $sup(s')=0$ otherwise; $sig(e)=\inp$, $sig(e')=\swap$ when $e'$ leads to or originates from the unique state with support $1$, and $\nop$ otherwise:
all the ESSAs $(e,s)$ will then be solved.

For the other events, we proceed as follows.
Let $v$ be some node of the graph $G$.

\medskip
If $v\not\in \Ss$, it may only occur in at most three paths of $B_G$ of the form\\
$T_i' =t_{i,0}\Edge{v}t_{i,1}\Edge{v_{i}}t_{i,2}\Edge{v}t_{i,3}\Edge{v_{i}'} t_{i,4}$ or $T_i' =t_{i,0}\Edge{v_{i}'}t_{i,1}\Edge{v}t_{i,2}\Edge{v_{i}}t_{i,3}\Edge{v} t_{i,4}$, since the companion vertex $v_i$ must be in $\Ss$; moreover, if $v$ belongs to several edges of $B_G$, all the corresponding companions must be different. Also, $v$ occurs twice, with a $v_i$ in between.
We may thus choose $sig(v)=\inp$, $sup(s)=1$ and $sup(s')=0$ if $s\edge{v}s'$, $sig(e)=\swap$ if $\edge{e}s$, in each path $T'_i$ containing $v$.
All the other events will have a signature $\nop$.
For any $T'_j$ not containing $v$, hence only containing signatures $\nop$ and $\swap$, the latter will introduce states with support $1$ while we would like to have supports $0$ to exclude $v$;
if this occurs, we shall thus consider two choices for the support of $t_{j,0}$: 0 (region $R_0$) and 1 (region $R_1$, with $sig(w_j)=\swap$); then the corresponding supports for the various $t_{j,\ell}$ will be complementary too.
In any case, we shall get regions separating $v$ from all the needed states.

\begin{figure}[!ht]
\vspace*{1mm}
\begin{minipage}{\textwidth}
\begin{center}
\scalebox{0.98}{
\begin{tikzpicture}[new set = import nodes]
\begin{scope}[nodes={set=import nodes}]
		\coordinate (bot1) at (0,0);
		\foreach \i in {0,...,4} {\coordinate (\i) at (\i*1.8cm+1.8cm,0);}
		\foreach \i in {1,3} {\fill[red!20] (\i) circle (0.3cm);}
		\node (bot1) at (bot1) {\nscale{$\bot_0$}};
		\draw[-latex](-0.6,0)to node[]{}(bot1);	
		\foreach \i in {0,...,4} {\node (\i) at (\i) {\nscale{$t_{0,\i}$}};}
		\graph {
	(import nodes);
			bot1->["\escale{$w_0$}"]0->["\escale{$v_0':\swap$}"]1->["\escale{$v_1:\inp$}"]2->["\escale{$v_0:\swap$}"]3->["\escale{$v_1:\inp$}"]4;
			
			};
\end{scope}
\begin{scope}[yshift=-1cm, nodes={set=import nodes}]
		\coordinate (bot2) at (0,0);
		\foreach \i in {0,...,4} {\coordinate (\i) at (\i*1.8cm+1.8cm,0);}
		\foreach \i in {1,2} {\fill[red!20] (\i) circle (0.3cm);}
		\node (bot2) at (bot2) {\nscale{$\bot_1$}};
		\foreach \i in {0,...,4} {\node (\i) at (\i) {\nscale{$t_{1,\i}$}};}			
		\graph {
	(import nodes);
			bot2->["\escale{$w_1$}"]0->["\escale{$v_0:\swap$}"]1->["\escale{$v_2$}"]2->["\escale{$v_0':\swap$}"]3->["\escale{$v_2'$}"]4;
			};
\end{scope}
\begin{scope}[yshift=-2cm, nodes={set=import nodes}]
		\coordinate (bot3) at (0,0);
		\foreach \i in {0,...,4} {\coordinate (\i) at (\i*1.8cm+1.8cm,0);}
		\foreach \i in {1,2} {\fill[red!20] (\i) circle (0.3cm);}
		\node (bot3) at (bot3) {\nscale{$\bot_2$}};
		\foreach \i in {0,...,4} {\node (\i) at (\i) {\nscale{$t_{2,\i}$}};}				
		\graph {
	(import nodes);
			bot3->["\escale{$w_2$}"]0->["\escale{$v_0':\swap$}"]1->["\escale{$v_3$}"]2->["\escale{$v_0:\swap$}"]3->["\escale{$v_3$}"]4;
			};
\end{scope}
\begin{scope}[yshift=-3cm, nodes={set=import nodes}]%
		\coordinate (bot4) at (0,0);
		\foreach \i in {0,...,4} {\coordinate (\i) at (\i*1.8cm+1.8cm,0);}
		\foreach \i in {0,2} {\fill[red!20] (\i) circle (0.3cm);}
		\node (bot4) at (bot4) {\nscale{$\bot_3$}};
		\foreach \i in {0,...,4} {\node (\i) at (\i) {\nscale{$t_{3,\i}$}};}					
		\graph {
	(import nodes);
			bot4->["\escale{$w_3:\swap$}"]0->["\escale{$v_1:\inp$}"]1->["\escale{$v_2:\swap$}"]2->["\escale{$v_1:\inp$}"]3->["\escale{$v_2'$}"]4;
			
			};
\end{scope}
\begin{scope}[yshift=-4cm,nodes={set=import nodes}]%
		\coordinate (bot5) at (0,0);
		\node (bot5) at (0,0) {\nscale{$\bot_4$}};
		\foreach \i in {0,...,4} { \coordinate (\i) at (1.8cm+ \i*1.8cm,0) ;}
		\foreach \i in {3,4} {\fill[red!20] (\i) circle (0.3cm);}
		\foreach \i in {0,...,4} { \node (\i) at (\i) {\nscale{$t_{4,\i}$}};}
		\graph {
	(import nodes);
			bot5->["\escale{$w_4$}"]0->["\escale{$v_2'$}"]1->["\escale{$v_3$}"]2->["\escale{$v_2:\swap$}"]3->["\escale{$v_3$}"]4;
			};
\end{scope}

\path (bot1) edge [->] node[left] {\nscale{$\ominus_1$} } (bot2);
\path (bot2) edge [->] node[left] {\nscale{$\ominus_2$} } (bot3);
\path (bot3) edge [->] node[left] {\nscale{$\ominus_3$} } (bot4);
\path (bot4) edge [->] node[left] {\nscale{$\ominus_4$} } (bot5);
\end{tikzpicture}
} \end{center}\vspace*{-5mm}
\caption{Region $R_0=(sup,sig)$ for $v=v_1$. $v$ occurs in $T'_0$ and $T'_3$. The supports of $t_{1,0},t_{2,0},t_{4,0}$ are chosen 0. The colored nodes have support 1, otherwise 0. When not indicated, signatures are $\nop$.}\label{fig:R_0}
\end{minipage}
\end{figure}
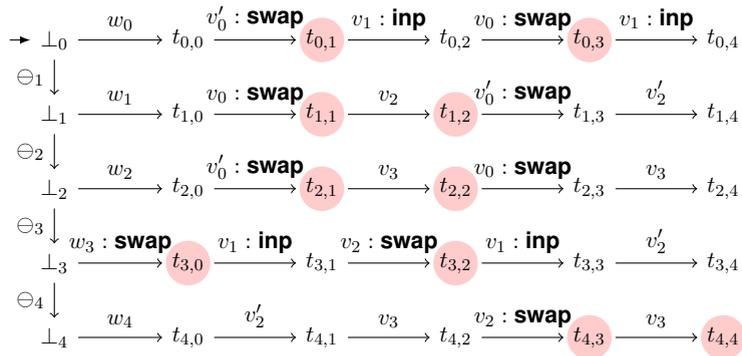

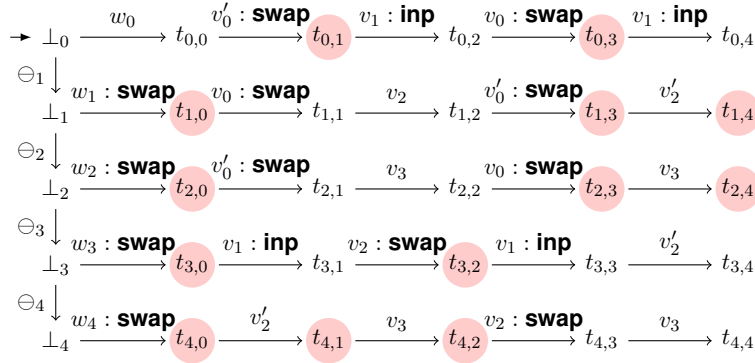
\begin{figure}[htbp]
\begin{minipage}{\textwidth}
\begin{center}
\begin{tikzpicture}[new set = import nodes]
\begin{scope}[nodes={set=import nodes}]
		\coordinate (bot1) at (0,0);
		\foreach \i in {0,...,4} {\coordinate (\i) at (\i*1.8cm+1.8cm,0);}
		\foreach \i in {1,3} {\fill[red!20] (\i) circle (0.3cm);}
		\node (bot1) at (bot1) {\nscale{$\bot_0$}};
		\draw[-latex](-0.6,0)to node[]{}(bot1);	
		\foreach \i in {0,...,4} {\node (\i) at (\i) {\nscale{$t_{0,\i}$}};}
		\graph {
	(import nodes);
			bot1->["\escale{$w_0$}"]0->["\escale{$v_0':\swap$}"]1->["\escale{$v_1:\inp$}"]2->["\escale{$v_0:\swap$}"]3->["\escale{$v_1:\inp$}"]4;
			
			};
\end{scope}
\begin{scope}[yshift=-1cm, nodes={set=import nodes}]
		\coordinate (bot2) at (0,0);
		\foreach \i in {0,...,4} {\coordinate (\i) at (\i*1.8cm+1.8cm,0);}
		\foreach \i in {0,3,4} {\fill[red!20] (\i) circle (0.3cm);}
		\node (bot2) at (bot2) {\nscale{$\bot_1$}};
		\foreach \i in {0,...,4} {\node (\i) at (\i) {\nscale{$t_{1,\i}$}};}			
		\graph {
	(import nodes);
			bot2->["\escale{$w_1:\swap$}"]0->["\escale{$v_0:\swap$}"]1->["\escale{$v_2$}"]2->["\escale{$v_0':\swap$}"]3->["\escale{$v_2'$}"]4;
			};
\end{scope}
\begin{scope}[yshift=-2cm, nodes={set=import nodes}]
		\coordinate (bot3) at (0,0);
		\foreach \i in {0,...,4} {\coordinate (\i) at (\i*1.8cm+1.8cm,0);}
		\foreach \i in {0,3,4} {\fill[red!20] (\i) circle (0.3cm);}
		\node (bot3) at (bot3) {\nscale{$\bot_2$}};
		\foreach \i in {0,...,4} {\node (\i) at (\i) {\nscale{$t_{2,\i}$}};}				
		\graph {
	(import nodes);
			bot3->["\escale{$w_2:\swap$}"]0->["\escale{$v_0':\swap$}"]1->["\escale{$v_3$}"]2->["\escale{$v_0:\swap$}"]3->["\escale{$v_3$}"]4;
			};
\end{scope}
\begin{scope}[yshift=-3cm, nodes={set=import nodes}]%
		\coordinate (bot4) at (0,0);
		\foreach \i in {0,...,4} {\coordinate (\i) at (\i*1.8cm+1.8cm,0);}
		\foreach \i in {0,2} {\fill[red!20] (\i) circle (0.3cm);}
		\node (bot4) at (bot4) {\nscale{$\bot_3$}};
		\foreach \i in {0,...,4} {\node (\i) at (\i) {\nscale{$t_{3,\i}$}};}					
		\graph {
	(import nodes);
			bot4->["\escale{$w_3:\swap$}"]0->["\escale{$v_1:\inp$}"]1->["\escale{$v_2:\swap$}"]2->["\escale{$v_1:\inp$}"]3->["\escale{$v_2'$}"]4;
			
			};
\end{scope}
\begin{scope}[yshift=-4cm,nodes={set=import nodes}]%
		\coordinate (bot5) at (0,0);
		\node (bot5) at (0,0) {\nscale{$\bot_4$}};
		\foreach \i in {0,...,4} { \coordinate (\i) at (1.8cm+ \i*1.8cm,0) ;}
		\foreach \i in {0,1,2} {\fill[red!20] (\i) circle (0.3cm);}
		\foreach \i in {0,...,4} { \node (\i) at (\i) {\nscale{$t_{4,\i}$}};}
		\graph {
	(import nodes);
			bot5->["\escale{$w_4:\swap$}"]0->["\escale{$v_2'$}"]1->["\escale{$v_3$}"]2->["\escale{$v_2:\swap$}"]3->["\escale{$v_3$}"]4;
			};
\end{scope}

\path (bot1) edge [->] node[left] {\nscale{$\ominus_1$} } (bot2);
\path (bot2) edge [->] node[left] {\nscale{$\ominus_2$} } (bot3);
\path (bot3) edge [->] node[left] {\nscale{$\ominus_3$} } (bot4);
\path (bot4) edge [->] node[left] {\nscale{$\ominus_4$} } (bot5);
\end{tikzpicture}
\end{center}
\caption{Region $R_1=(sup,sig)$ for $v=v_1$. $v$ occurs in $T'_0$ and $T'_3$. The supports of $t_{1,0},t_{2,0},t_{4,0}$ are chosen 1. The colored nodes have support 1, otherwise 0. When not indicated, signatures are $\nop$.}\label{fig:R_1}
\end{minipage}
\end{figure}

If $v\in \Ss$, it may only occur in at most three paths of $B_G$ of the form \\
$T_i' =t_{i,0}\Edge{v'}t_{i,1}\Edge{v_{i}}t_{i,2}\Edge{v}t_{i,3}\Edge{v_{i}} t_{i,4}$ \ or \
 $T_i' =t_{i,0}\Edge{v_{i}}t_{i,1}\Edge{v}t_{i,2}\Edge{v_{i}}t_{i,3}\Edge{v'} t_{i,4}$,
 \\
 or $T_i' =t_{i,0}\Edge{v}t_{i,1}\Edge{v_{i}}t_{i,2}\Edge{v'}t_{i,3}\Edge{v'_{i}} t_{i,4}$ or
 $T_i' =t_{i,0}\Edge{v_{i}}t_{i,1}\Edge{v}t_{i,2}\Edge{v'_{i}}t_{i,3}\Edge{v'} t_{i,4}$,
  depending on the fact that the companion vertex $v_i$ belongs to $\Ss$ or not.
Again, if $v$ belongs to several edges, all the corresponding companions must be different.
Also, $v$ occurs only once, as well as $v'$.
For $v$, we may thus choose $sig(v)=\inp$, $sup(s)=1$ and $sup(s')=0$ if $s\edge{v}s'$, $sig(e)=\swap$ if $\edge{e}s$, in each path $T'_i$ containing $v$.
All the other events will have a signature.
For any $T'_j$ not containing $v$, we may have some events with signature $\swap$, hence some states with support $1$ while we would like to only have supports $0$ in order to exclude $v$;
hence, if needed, we shall again consider two choices for the support of $t_{j,0}$: 0 (region $R_2$) and 1 (region $R_3$);
then the corresponding supports for the various $t_{j,\ell}$ will be complementary too.
In any case, we shall get regions separating $v$ from all the needed states.
However, for the first two kinds of configuration, since $v$ is between two $v_i$'s, after the second $v_i$ (with signature $\swap$), we shall have states with support $1$ after $v$, while we would like to have supports $0$ to exclude $v$;
hence we will also use a region where $w_i$ has signature $\swap$ and $v_i$ has signature $\nop$ (region $R_3$).
Note that, in this case, no event in any $T'_j$ not containing $v$ has a signature $\swap$ and all the states have support $0$;
hence we do not need here to introduce an additional region with $sig(w_j)=\swap$ to get complementary supports.
In any case, we shall get regions separating $v$ from all the needed states.

\medskip
For $v'$, we proceed similarly.

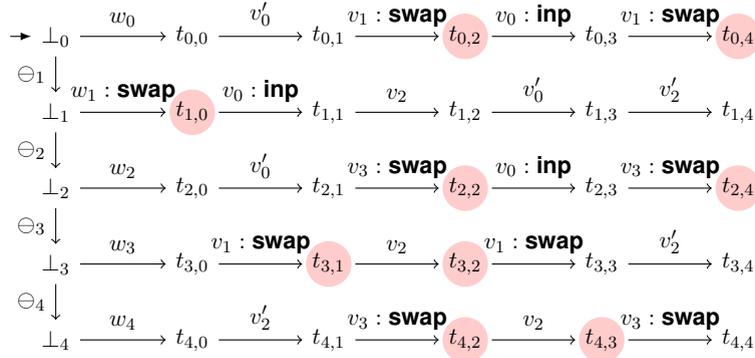
\begin{figure}[htbp]
\begin{minipage}{\textwidth}
\begin{center}
\begin{tikzpicture}[new set = import nodes]
\begin{scope}[nodes={set=import nodes}]
		\coordinate (bot1) at (0,0);
		\foreach \i in {0,...,4} {\coordinate (\i) at (\i*1.8cm+1.8cm,0);}
		\foreach \i in {2,4} {\fill[red!20] (\i) circle (0.3cm);}
		\node (bot1) at (bot1) {\nscale{$\bot_0$}};
		\draw[-latex](-0.6,0)to node[]{}(bot1);	
		\foreach \i in {0,...,4} {\node (\i) at (\i) {\nscale{$t_{0,\i}$}};}
		\graph {
	(import nodes);
			bot1->["\escale{$w_0$}"]0->["\escale{$v_0'$}"]1->["\escale{$v_1:\swap$}"]2->["\escale{$v_0:\inp$}"]3->["\escale{$v_1:\swap$}"]4;
			
			};
\end{scope}
\begin{scope}[yshift=-1cm, nodes={set=import nodes}]
		\coordinate (bot2) at (0,0);
		\foreach \i in {0,...,4} {\coordinate (\i) at (\i*1.8cm+1.8cm,0);}
		\foreach \i in {0} {\fill[red!20] (\i) circle (0.3cm);}
		\node (bot2) at (bot2) {\nscale{$\bot_1$}};
		\foreach \i in {0,...,4} {\node (\i) at (\i) {\nscale{$t_{1,\i}$}};}			
		\graph {
	(import nodes);
			bot2->["\escale{$w_1:\swap$}"]0->["\escale{$v_0:\inp$}"]1->["\escale{$v_2$}"]2->["\escale{$v_0'$}"]3->["\escale{$v_2'$}"]4;
			};
\end{scope}
\begin{scope}[yshift=-2cm, nodes={set=import nodes}]
		\coordinate (bot3) at (0,0);
		\foreach \i in {0,...,4} {\coordinate (\i) at (\i*1.8cm+1.8cm,0);}
		\foreach \i in {2,4} {\fill[red!20] (\i) circle (0.3cm);}
		\node (bot3) at (bot3) {\nscale{$\bot_2$}};
		\foreach \i in {0,...,4} {\node (\i) at (\i) {\nscale{$t_{2,\i}$}};}				
		\graph {
	(import nodes);
			bot3->["\escale{$w_2$}"]0->["\escale{$v_0'$}"]1->["\escale{$v_3:\swap$}"]2->["\escale{$v_0:\inp$}"]3->["\escale{$v_3:\swap$}"]4;
			};
\end{scope}
\begin{scope}[yshift=-3cm, nodes={set=import nodes}]%
		\coordinate (bot4) at (0,0);
		\foreach \i in {0,...,4} {\coordinate (\i) at (\i*1.8cm+1.8cm,0);}
		\foreach \i in {1,2} {\fill[red!20] (\i) circle (0.3cm);}
		\node (bot4) at (bot4) {\nscale{$\bot_3$}};
		\foreach \i in {0,...,4} {\node (\i) at (\i) {\nscale{$t_{3,\i}$}};}					
		\graph {
	(import nodes);
			bot4->["\escale{$w_3$}"]0->["\escale{$v_1:\swap$}"]1->["\escale{$v_2$}"]2->["\escale{$v_1:\swap$}"]3->["\escale{$v_2'$}"]4;
			
			};
\end{scope}
\begin{scope}[yshift=-4cm,nodes={set=import nodes}]%
		\coordinate (bot5) at (0,0);
		\node (bot5) at (0,0) {\nscale{$\bot_4$}};
		\foreach \i in {0,...,4} { \coordinate (\i) at (1.8cm+ \i*1.8cm,0) ;}
		\foreach \i in {2,3} {\fill[red!20] (\i) circle (0.3cm);}
		\foreach \i in {0,...,4} { \node (\i) at (\i) {\nscale{$t_{4,\i}$}};}
		\graph {
	(import nodes);
			bot5->["\escale{$w_4$}"]0->["\escale{$v_2'$}"]1->["\escale{$v_3:\swap$}"]2->["\escale{$v_2$}"]3->["\escale{$v_3:\swap$}"]4;
			};
\end{scope}

\path (bot1) edge [->] node[left] {\nscale{$\ominus_1$} } (bot2);
\path (bot2) edge [->] node[left] {\nscale{$\ominus_2$} } (bot3);
\path (bot3) edge [->] node[left] {\nscale{$\ominus_3$} } (bot4);
\path (bot4) edge [->] node[left] {\nscale{$\ominus_4$} } (bot5);
\end{tikzpicture}
\end{center}
\caption{Region $R_2=(sup,sig)$ for $v=v_0$. $v$ occurs in $T'_0$, $T'_1$ and $T'_2$. The supports of $t_{3,0},t_{4,0}$ are chosen 0.
The colored nodes have support 1, otherwise 0. When not indicated, signatures are $\nop$.}\label{fig:R_2}
\end{minipage}
\end{figure}

\begin{figure}[htbp]
\begin{minipage}{\textwidth}
\begin{center}
\begin{tikzpicture}[new set = import nodes]
\begin{scope}[nodes={set=import nodes}]
		\coordinate (bot1) at (0,0);
		\foreach \i in {0,...,4} {\coordinate (\i) at (\i*1.8cm+1.8cm,0);}
		\foreach \i in {0,1,2} {\fill[red!20] (\i) circle (0.3cm);}
		\node (bot1) at (bot1) {\nscale{$\bot_0$}};
		\draw[-latex](-0.6,0)to node[]{}(bot1);	
		\foreach \i in {0,...,4} {\node (\i) at (\i) {\nscale{$t_{0,\i}$}};}
		\graph {
	(import nodes);
			bot1->["\escale{$w_0:\swap$}"]0->["\escale{$v_0'$}"]1->["\escale{$v_1$}"]2->["\escale{$v_0:\inp$}"]3->["\escale{$v_1$}"]4;
			
			};
\end{scope}
\begin{scope}[yshift=-1cm, nodes={set=import nodes}]
		\coordinate (bot2) at (0,0);
		\foreach \i in {0,...,4} {\coordinate (\i) at (\i*1.8cm+1.8cm,0);}
		\foreach \i in {0} {\fill[red!20] (\i) circle (0.3cm);}
		\node (bot2) at (bot2) {\nscale{$\bot_1$}};
		\foreach \i in {0,...,4} {\node (\i) at (\i) {\nscale{$t_{1,\i}$}};}			
		\graph {
	(import nodes);
			bot2->["\escale{$w_1:\swap$}"]0->["\escale{$v_0:\inp$}"]1->["\escale{$v_2$}"]2->["\escale{$v_0'$}"]3->["\escale{$v_2'$}"]4;
			};
\end{scope}
\begin{scope}[yshift=-2cm, nodes={set=import nodes}]
		\coordinate (bot3) at (0,0);
		\foreach \i in {0,...,4} {\coordinate (\i) at (\i*1.8cm+1.8cm,0);}
		\foreach \i in {0,1,2} {\fill[red!20] (\i) circle (0.3cm);}
		\node (bot3) at (bot3) {\nscale{$\bot_2$}};
		\foreach \i in {0,...,4} {\node (\i) at (\i) {\nscale{$t_{2,\i}$}};}				
		\graph {
	(import nodes);
			bot3->["\escale{$w_2:\swap$}"]0->["\escale{$v_0'$}"]1->["\escale{$v_3$}"]2->["\escale{$v_0:\inp$}"]3->["\escale{$v_3$}"]4;
			};
\end{scope}
\begin{scope}[yshift=-3cm, nodes={set=import nodes}]%
		\coordinate (bot4) at (0,0);
		\foreach \i in {0,...,4} {\coordinate (\i) at (\i*1.8cm+1.8cm,0);}
		\node (bot4) at (bot4) {\nscale{$\bot_3$}};
		\foreach \i in {0,...,4} {\node (\i) at (\i) {\nscale{$t_{3,\i}$}};}					
		\graph {
	(import nodes);
			bot4->["\escale{$w_3$}"]0->["\escale{$v_1$}"]1->["\escale{$v_2$}"]2->["\escale{$v_1$}"]3->["\escale{$v_2'$}"]4;
			
			};
\end{scope}
\begin{scope}[yshift=-4cm,nodes={set=import nodes}]%
		\coordinate (bot5) at (0,0);
		\node (bot5) at (0,0) {\nscale{$\bot_4$}};
		\foreach \i in {0,...,4} { \coordinate (\i) at (1.8cm+ \i*1.8cm,0) ;}
		\foreach \i in {0,...,4} { \node (\i) at (\i) {\nscale{$t_{4,\i}$}};}
		\graph {
	(import nodes);
			bot5->["\escale{$w_4$}"]0->["\escale{$v_2'$}"]1->["\escale{$v_3$}"]2->["\escale{$v_2$}"]3->["\escale{$v_3$}"]4;
			};
\end{scope}

\path (bot1) edge [->] node[left] {\nscale{$\ominus_1$} } (bot2);
\path (bot2) edge [->] node[left] {\nscale{$\ominus_2$} } (bot3);
\path (bot3) edge [->] node[left] {\nscale{$\ominus_3$} } (bot4);
\path (bot4) edge [->] node[left] {\nscale{$\ominus_4$} } (bot5);
\end{tikzpicture}
\end{center}\vspace*{-3mm}
\caption{Region $R_3=(sup,sig)$ for $v=v_0$.
$v$ occurs in $T'_0$, $T'_1$ and $T'_2$.
The supports of $t_{3,0},t_{4,0}$ are chosen 0.
The colored nodes have support 1, otherwise 0.
When not indicated, signatures are $\nop$.
Note that here no $v_i$ with $i\neq 0$ needs to have a signature $\swap$;
hence we do not need to add a region with signature $\swap$ for $w_j$ when $T'_j$ does not contain $v$.}\label{fig:R_3}
\end{minipage}
\end{figure}
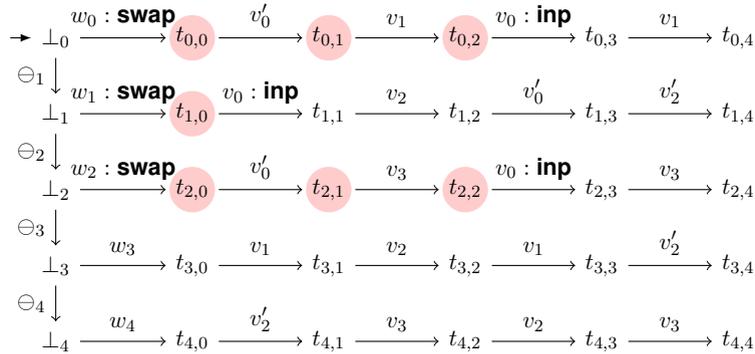

This proves the lemma.
\end{proof}
%

\subsection{The proof of Theorem~\ref{the:label_splitting}(\ref{the:label_splitting_ls_and_real}), when  $\tau\cap\{\inp,\out\}=\emptyset$}\label{sec:split-used}%

Since we already handled the case $\tau=\{\nop,\swap\}$, we may assume $\tau=\{\nop,\swap\}\cup\omega$ with $\emptyset\not=\omega\subseteq \{\used,\free\}$.
Then, if we have to solve an ESSA $(e,s)$, we have to try to find a $\tau$-region where $sig(e)=\used$ and $\sup(s)=0$, or  $sig(e)=\free$ and $\sup(s)=1$, but if $s'\edge{e}s\edge{\neg e}$, in either case $sup(s')=sup(s)$ and we may not solve $(e,s)$.

\medskip
Thus, there is no $E'$-label-splitting of the TS $A_G$ of Section~\ref{sec:label_splitting_ssp}, that has the $\tau$-ESSP.
To overcome this obstacle, with as little effort as possible, we shall use transition systems $\overline{A}_G$ and $\overline{B}_G$ which extend $A_G$ and $B_G$ by backward-edges (see Figure~\ref{fig:overline A and B}).
Similarly, we shall denote by $\overline{T}_i$ the subsystem $T_i$ with backward edges, for any $i$.

\medskip

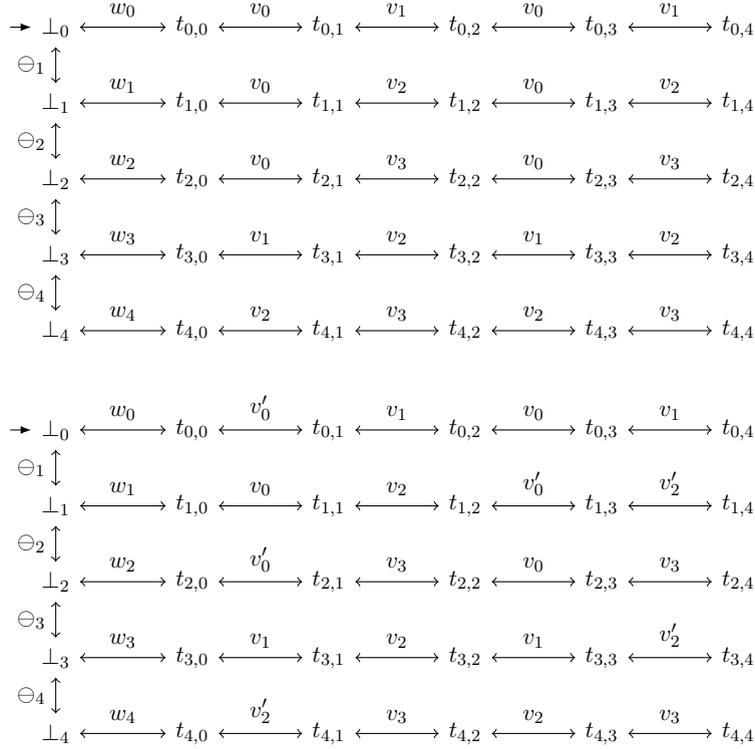
\begin{figure}[h!]
\begin{center}
\begin{tikzpicture}[new set = import nodes]
\begin{scope}[nodes={set=import nodes}]
		\coordinate (bot1) at (0,0);
		\foreach \i in {0,...,4} {\coordinate (\i) at (\i*1.8cm+1.8cm,0);}
		\node (bot1) at (bot1) {\nscale{$\bot_0$}};
		\draw[-latex](-0.6,0)to node[]{}(bot1);	
		\foreach \i in {0,...,4} {\node (\i) at (\i) {\nscale{$t_{0,\i}$}};}
		\graph {
	(import nodes);
			bot1<->["\escale{$w_0$}"]0<->["\escale{$v_0$}"]1<->["\escale{$v_1$}"]2<->["\escale{$v_0$}"]3<->["\escale{$v_1$}"]4;
			
			};
\end{scope}
\begin{scope}[yshift=-1cm, nodes={set=import nodes}]
		\coordinate (bot2) at (0,0);
		\foreach \i in {0,...,4} {\coordinate (\i) at (\i*1.8cm+1.8cm,0);}
		\node (bot2) at (bot2) {\nscale{$\bot_1$}};
		\foreach \i in {0,...,4} {\node (\i) at (\i) {\nscale{$t_{1,\i}$}};}			
		\graph {
	(import nodes);
			bot2<->["\escale{$w_1$}"]0<->["\escale{$v_0$}"]1<->["\escale{$v_2$}"]2<->["\escale{$v_0$}"]3<->["\escale{$v_2$}"]4;
			};
\end{scope}
\begin{scope}[yshift=-2cm, nodes={set=import nodes}]
		\coordinate (bot3) at (0,0);
		\foreach \i in {0,...,4} {\coordinate (\i) at (\i*1.8cm+1.8cm,0);}
		\node (bot3) at (bot3) {\nscale{$\bot_2$}};
		\foreach \i in {0,...,4} {\node (\i) at (\i) {\nscale{$t_{2,\i}$}};}				
		\graph {
	(import nodes);
			bot3<->["\escale{$w_2$}"]0<->["\escale{$v_0$}"]1<->["\escale{$v_3$}"]2<->["\escale{$v_0$}"]3<->["\escale{$v_3$}"]4;
			};
\end{scope}
\begin{scope}[yshift=-3cm, nodes={set=import nodes}]%
		\coordinate (bot4) at (0,0);
		\foreach \i in {0,...,4} {\coordinate (\i) at (\i*1.8cm+1.8cm,0);}
		\node (bot4) at (bot4) {\nscale{$\bot_3$}};
		\foreach \i in {0,...,4} {\node (\i) at (\i) {\nscale{$t_{3,\i}$}};}					
		\graph {
	(import nodes);
			bot4<->["\escale{$w_3$}"]0<->["\escale{$v_1$}"]1<->["\escale{$v_2$}"]2<->["\escale{$v_1$}"]3<->["\escale{$v_2$}"]4;
			
			};
\end{scope}
\begin{scope}[yshift=-4cm,nodes={set=import nodes}]%
		\coordinate (bot5) at (0,0);
		\node (bot5) at (0,0) {\nscale{$\bot_4$}};
		\foreach \i in {0,...,4} { \coordinate (\i) at (1.8cm+ \i*1.8cm,0) ;}
		\foreach \i in {0,...,4} { \node (\i) at (\i) {\nscale{$t_{4,\i}$}};}
		\graph {
	(import nodes);
			bot5<->["\escale{$w_4$}"]0<->["\escale{$v_2$}"]1<->["\escale{$v_3$}"]2<->["\escale{$v_2$}"]3<->["\escale{$v_3$}"]4;
			};
\end{scope}

\path (bot1) edge [<->] node[left] {\nscale{$\ominus_1$} } (bot2);
\path (bot2) edge [<->] node[left] {\nscale{$\ominus_2$} } (bot3);
\path (bot3) edge [<->] node[left] {\nscale{$\ominus_3$} } (bot4);
\path (bot4) edge [<->] node[left] {\nscale{$\ominus_4$} } (bot5);
\end{tikzpicture}
\end{center}
\begin{center}
\begin{tikzpicture}[new set = import nodes]
\begin{scope}[nodes={set=import nodes}]
		\coordinate (bot1) at (0,0);
		\foreach \i in {0,...,4} {\coordinate (\i) at (\i*1.8cm+1.8cm,0);}
		\node (bot1) at (bot1) {\nscale{$\bot_0$}};
		\draw[-latex](-0.6,0)to node[]{}(bot1);	
		\foreach \i in {0,...,4} {\node (\i) at (\i) {\nscale{$t_{0,\i}$}};}
		\graph {
	(import nodes);
			bot1<->["\escale{$w_0$}"]0<->["\escale{$v_0'$}"]1<->["\escale{$v_1$}"]2<->["\escale{$v_0$}"]3<->["\escale{$v_1$}"]4;
			
			};
\end{scope}
\begin{scope}[yshift=-1cm, nodes={set=import nodes}]
		\coordinate (bot2) at (0,0);
		\foreach \i in {0,...,4} {\coordinate (\i) at (\i*1.8cm+1.8cm,0);}
		\node (bot2) at (bot2) {\nscale{$\bot_1$}};
		\foreach \i in {0,...,4} {\node (\i) at (\i) {\nscale{$t_{1,\i}$}};}			
		\graph {
	(import nodes);
			bot2<->["\escale{$w_1$}"]0<->["\escale{$v_0$}"]1<->["\escale{$v_2$}"]2<->["\escale{$v_0'$}"]3<->["\escale{$v_2'$}"]4;
			};
\end{scope}
\begin{scope}[yshift=-2cm, nodes={set=import nodes}]
		\coordinate (bot3) at (0,0);
		\foreach \i in {0,...,4} {\coordinate (\i) at (\i*1.8cm+1.8cm,0);}
		\node (bot3) at (bot3) {\nscale{$\bot_2$}};
		\foreach \i in {0,...,4} {\node (\i) at (\i) {\nscale{$t_{2,\i}$}};}				
		\graph {
	(import nodes);
			bot3<->["\escale{$w_2$}"]0<->["\escale{$v_0'$}"]1<->["\escale{$v_3$}"]2<->["\escale{$v_0$}"]3<->["\escale{$v_3$}"]4;
			};
\end{scope}
\begin{scope}[yshift=-3cm, nodes={set=import nodes}]%
		\coordinate (bot4) at (0,0);
		\foreach \i in {0,...,4} {\coordinate (\i) at (\i*1.8cm+1.8cm,0);}
		\node (bot4) at (bot4) {\nscale{$\bot_3$}};
		\foreach \i in {0,...,4} {\node (\i) at (\i) {\nscale{$t_{3,\i}$}};}					
		\graph {
	(import nodes);
			bot4<->["\escale{$w_3$}"]0<->["\escale{$v_1$}"]1<->["\escale{$v_2$}"]2<->["\escale{$v_1$}"]3<->["\escale{$v_2'$}"]4;
			
			};
\end{scope}
\begin{scope}[yshift=-4cm,nodes={set=import nodes}]%
		\coordinate (bot5) at (0,0);
		\node (bot5) at (0,0) {\nscale{$\bot_4$}};
		\foreach \i in {0,...,4} { \coordinate (\i) at (1.8cm+ \i*1.8cm,0) ;}
		\foreach \i in {0,...,4} { \node (\i) at (\i) {\nscale{$t_{4,\i}$}};}
		\graph {
	(import nodes);
			bot5<->["\escale{$w_4$}"]0<->["\escale{$v_2'$}"]1<->["\escale{$v_3$}"]2<->["\escale{$v_2$}"]3<->["\escale{$v_3$}"]4;
			};
\end{scope}

\path (bot1) edge [<->] node[left] {\nscale{$\ominus_1$} } (bot2);
\path (bot2) edge [<->] node[left] {\nscale{$\ominus_2$} } (bot3);
\path (bot3) edge [<->] node[left] {\nscale{$\ominus_3$} } (bot4);
\path (bot4) edge [<->] node[left] {\nscale{$\ominus_4$} } (bot5);
\end{tikzpicture}
\end{center}\vspace*{-3mm}
\caption{ \label{fig:overline A and B}
The new transition systems  $\overline{A}_G$ (top) and  $\overline{B}_G$ (bottom) that originate from Example~\ref{ex:vc}.}
\end{figure}

\begin{lemma}\label{lem:essp_implies_vc_2}
If there is an $E'$-label-splitting of $\overline{A}_G$ such that $\vert E'\vert\leq \kappa$ that has the $\tau$-ESSP, then $G$ has a $\lambda$-VC.
\end{lemma}
\begin{proof}
The proof is similar to the one of Lemma~\ref{lem:essp_implies_vc}.
\end{proof}
Conversely, let $\Ss$ be a $\lambda$-VC of $G$, and let $\overline{B}_G=(S(B_G),E', \delta'', \bot_0)$ be the bi-directed extension of the TS $B_G=(S(B_G), E',\delta',\bot_0)$, which has been defined in Section~\ref{sec:label_splitting_ssp}.
That is, for all $s,s'\in S(B_G)$ and all $e\in E'$ if $\delta'(s,e)=s'$, then $\delta''(s,e)=s'$ and $\delta''(s',e)=s$.
To complete the proof of Theorem~\ref{the:label_splitting}.2 it remains to argue that $\overline{B}_G$ has the $\tau$-ESSP and the $\tau$-SSP.
Recall that the signatures of the regions that have been presented for the proof of Lemma~\ref{lem:vc_implies_ssp} only use $\nop$ and $\swap$.
Thus, they can be directly applied to $\overline{B}_G$, which proves $\overline{B}_G$'s $\tau$-SSP.
Hence, it remains to argue for $\overline{B}_G$'s ESSP.
The following lemma confirms both properties for $\overline{B}_G$ and thus completes the proof of Theorem~\ref{the:label_splitting}.
\begin{lemma}\label{lem:vc_implies_essp_2}
The TS $\overline{B}_G$ has the $\tau$-SSP, and the $\tau$-ESSP.
\end{lemma}
\begin{proof}
First, we may observe that the signatures of the regions that have been presented for the proof of Lemma~\ref{lem:vc_implies_ssp} only use $\nop$ and $\swap$.
Thus, they can be directly applied to $\overline{B}_G$, which proves $\overline{B}_G$'s $\tau$-SSP.

\medskip
We now argue that $\overline{B}_G$ has the $\tau$-ESSP.
Without loss of generality, we assume $\used\in \tau$.
(Otherwise $\free\in \tau$ and this case is similar.)
We thus have to find, for each ESSA $(e,s)$, a region $R=(sup,sig)$ such that $sig(e)=\used$ and $sup(s)=0$.
The proof is very similar to the one for Lemma~\ref{lem:vc_implies_essp}.

For instance, if $e\in W \cup\ominus$, one may choose $sup(s')=1$ if $s'\edge{e}$ and $sup(s')=0$ otherwise;
$sig(e)=\used$, $sig(e')=\swap$ when $e'$ leads to or originates from the unique state with support $1$, and $\nop$ otherwise:
all the ESSAs $(e,s)$ will then be solved.
For the other events, we proceed as follows.

\medskip
Let $v$ be some node of the graph $G$.\\
If $v\not\in \Ss$, it may only occur in at most three edges of $G$ (in each case, the companion vertex must be in $\Ss$, and if $v$ belongs to several edges of $G$, all the corresponding companions must be different), which leads to the two kinds of decorated paths in $\overline{B}_G$:

\medskip
$\bot_i:0\FBedge{w_i:\swap}t_{i,0}:1\FBedge{v:\used}t_{i,1}:1\FBedge{v_{i}:\nop}t_{i,2}:1\FBedge{v:\used}t_{i,3}:1\FBedge{v_{i}':\swap} t_{i,4}:0$ \\
or $t_{i,0}:0\FBedge{v_{i}':\swap}t_{i,1}:1\FBedge{v:\used}t_{i,2}:1\FBedge{v_{i}:\nop}t_{i,3}:1\FBedge{v:\used} t_{i,4}:1$.

\medskip\noindent
For any $\overline{T}_j$ not containing $v$, hence only containing signatures $\nop$ and $\swap$, this may introduce states with support $1$ while we would like to have supports $0$ to exclude $v$;
if this occurs, we may consider two choices for the support of $t_{j,0}$: 0 (region $R_0$) and 1 (region $R_1$, with $sig(w_j)=\swap$);
then the corresponding supports for the various $t_{j,\ell}$ will be complementary too.
In any case, we shall get regions separating $v$ from all the needed states.

\medskip
If $v\in \Ss$, this leads us to four kinds of decorated paths in $\overline{B}_G$, depending on the fact that the companion vertex $v_i$ belongs to $\Ss$ or not, and if $v$ is smaller than $v_i$ or not:

\medskip\noindent
$t_{i,0}:0\FBedge{v':\nop}t_{i,1}:0\FBedge{v_{i}:\swap}t_{i,2}:1\FBedge{v:\used}t_{i,3}:1\FBedge{v_{i}:\swap} t_{i,4}:0$\\
 or $t_{i,0}:0\FBedge{v_{i}:\swap}t_{i,1}:1\FBedge{v:\used}t_{i,2}:1\FBedge{v_{i}:\swap}t_{i,3}:0\FBedge{v':\nop} t_{i,4}:0$\\
or $\bot_i:0\FBedge{w_i:\swap}t_{i,0}:1\FBedge{v:\used}t_{i,1}:1\FBedge{v_{i}:\swap}t_{i,2}:0\FBedge{v':\nop}t_{i,3}:0\FBedge{v'_{i}:\nop} t_{i,4}:0$\\
or $t_{i,0}:0\FBedge{v_{i}:\swap}t_{i,1}:1\FBedge{v:\used}t_{i,2}:1\FBedge{v'_{i}:\swap}t_{i,3}:0\FBedge{v':\nop} t_{i,4}:0$.

\medskip\noindent
Again, if $v$ belongs to several edges, all the corresponding companions must be different.
Also, $v$ occurs only once, as well as $v'$.
For any $\overline{T}_j$ not containing $v$, we may have some events with signature $\swap$ (the other ones having $\nop$), hence some states with support $1$ while we would like to only have supports $0$ in order to exclude $v$;
hence, if needed, we shall again consider two choices for the support of $t_{j,0}$: 0 (region $R_2$) and 1(region $R_3$);
then the corresponding supports for the various $t_{j,\ell}$ will be complementary too.
In any case, we shall get regions separating $v$ from all the needed states.
\\ For $v'$, we proceed similarly.

Altogether, by the arbitrariness of $v$, we proved the $\tau$-ESSP for $\overline{B}_G$ and this thus completes the proof of Lemma~\ref{lem:vc_implies_essp_2}.
\end{proof}

\section{The complexity of  edge-removal}\label{sec:edge_removal}%

In order to make a TS implementable, i.e., to satisfy SSP and/or ESSP, the removal of edges can also be an appropriate way of modification:

\begin{definition}[Edge-Removal]\label{def:edge_removal}
Let $A=(S,E,\delta,\iota)$ be a TS.
A TS\  $B=(S, E,\delta',\iota)$ is an \emph{edge-removal} of $A$ if, for all $e\in E'$ and all $s,s'\in S'$, holds: if $s\edge{e}s'\in B$, then $s\edge{e}s'\in A$.
By $\K=\{s\edge{e}s'\in A\mid s\edge{e}s'\not\in B\}$ we refer to the (set of) removed edges.
\end{definition}

We would like to emphasize that $B$ and $A$ have the same set of states and events.
Moreover, $B$ is assumed to be a valid system, i.e., each state remains reachable from the initial one and each event occurs at least once in $\delta'$.

\noindent
\fbox{\begin{minipage}[t][1.8\height][c]{0.88\textwidth}
\begin{decisionproblem}
  \problemtitle{\textsc{$\tau$-Edge-Removal for Embedding}}
  \probleminput{A TS\  $A=(S,E,\delta,\iota)$, a natural number $\kappa$.}
  \problemquestion{Does there exist an edge-removal $B$ for $A$ by $\K$ that has the $\tau$-SSP and satisfies $\vert\K\vert\leq \kappa$?}
\end{decisionproblem}
\end{minipage}}
\smallskip

\noindent
\fbox{\begin{minipage}[t][1.8\height][c]{0.88\textwidth}
\begin{decisionproblem}
  \problemtitle{\textsc{$\tau$-Edge-Removal for Language-Simulation}}
  \probleminput{A TS\  $A=(S,E,\delta,\iota)$, a natural number $\kappa$.}
  \problemquestion{Does there exist an edge-removal $B$ for $A$ by $\K$ that has the $\tau$-ESSP and satisfies $\vert\K\vert\leq \kappa$?}
\end{decisionproblem}
\end{minipage}}
\smallskip

\noindent
\fbox{\begin{minipage}[t][1.8\height][c]{0.88\textwidth}
\begin{decisionproblem}
  \problemtitle{\textsc{$\tau$-Edge-Removal for Realization}}
  \probleminput{A TS\  $A=(S,E,\delta,\iota)$, a natural number $\kappa$.}
  \problemquestion{Does there exist an edge-removal $B$ for $A$ by $\K$ that has the $\tau$-ESSP and the $\tau$-SSP and satisfies $\vert\K\vert\leq \kappa$?}
\end{decisionproblem}
\end{minipage}}

\medskip

The following theorem characterizes the complexity of the edge-removal problem for all implementations and types under consideration:

\begin{theorem}\label{the:edge_removal}
If $\omega\subseteq\{\inp,\out, \free,\used\}$, and $\tau=\{\nop,\swap\}\cup\omega$, then
\begin{enumerate}
\itemsep=0.95pt
\item\label{the:edge_removal_embedding}
\textsc{$\tau$-Edge-Removal for Embedding} is NP-complete.
\item\label{the:edge_removal_langsim_real}
\textsc{$\tau$-Edge-Removal for Language-Simulation} and \\
\textsc{$\tau$-Edge-Removal for Realization} are NP-complete if $\omega\not=\emptyset$, otherwise they are solvable in polynomial time.
\end{enumerate}
\end{theorem}

First of all, we argue for the polynomial part:
If $\tau=\{\nop,\swap\}$, then a TS $A=(S,E,\delta,\iota)$ has the $\tau$-ESSP if and only if every event occurs at every state,
since the functions \nop,\swap{} are defined on both $0$ and $1$.
Thus, any ESSA $(e,s)$ of $A$ would be unsolvable.
This implies that an edge-removal may neither render nor keep the $\tau$-ESSP valid, since the removal of an edge would produce an unsolvable ESSA.
Hence, the decision problems are polynomial, since either $A$ is already implementable, which can be checked in polynomial time~\cite{tamc/TredupR19}, or it has to be rejected.
Thus, for the proof of Theorem~\ref{the:edge_removal}, it remains to consider the NP-completeness results.

\medskip
In order to prove Theorem~\ref{the:edge_removal},  we present suitable reductions of \text{3BVC}, where we reduce an input $G=(\U,M)$ to an instance $(A_G,\kappa)$, such that $G$ has a
$\lambda$-VC if and only if $A_G$ allows an implementable edge-removal $B_G$ that satisfies $\vert \K\vert\leq \kappa$.
However, due to their different ability to solve ESSAs, when it comes to language-simulation or realization, we have to distinguish again between the types $\tau$ that have at least one of $\inp$ or $\out$, and the ones that do not have any of them.

\subsection{The proof of Theorem~\ref{the:edge_removal}(\ref{the:edge_removal_embedding}), and the proof of  Theorems~\ref{the:edge_removal}(\ref{the:edge_removal_langsim_real}) for the types with \inp\ or \out}\label{sec:edge_removal_inp_out}

In this section, we shall prove Theorem~\ref{the:edge_removal}(\ref{the:edge_removal_embedding}) for all types $\tau=\{\nop,\swap\}\cup\omega$ with $\omega\subseteq\{\inp, \out, \free, \used\}$, and we prove Theorems~\ref{the:edge_removal}(\ref{the:edge_removal_langsim_real}) for the types that additionally satisfy $\omega\cap\{\inp,\out\}\not=\emptyset$.
We deal with these proofs simultaneously, since they use the same reduction.
In particular, according to our general approach, we start from an input $(G, \lambda)$ of \textsc{3BVC}, and construct an instance TS $(A_G, \kappa)$ as follows: \medskip\\
If $\tau=\{\nop,\swap\}\cup\omega$ with $\omega\subseteq\{\inp, \out, \free, \used\}$, then $A_G$ allows an edge-removal $B_G$ that respects $\kappa$, and has the $\tau$-SSP if and only if $G$ has a $\lambda$-vertex-cover; \medskip\\
If $\tau=\{\nop,\swap\}\cup\omega$ with $\omega\subseteq\{\inp, \out, \free, \used\}$, and $\omega\cap\{\inp,\out\}\not=\emptyset$, then $A_G$ allows an edge-removal $B_G$ that respects $\kappa$, and has the $\tau$-ESSP if and only if $G$ has a $\lambda$-vertex-cover.

\medskip
Hence, in the remainder of this section,  let $\tau=\{\nop,\swap\}\cup\omega$ with $\omega\subseteq\{\inp, \out, \free, \used\}$, whenever we deal with the $\tau$-SSP,  and let additionally $\omega\cap\{\inp,\out\}\not=\emptyset$, whenever we deal with the $\tau$-ESSP (where the TS in question is clear from the context).

\medskip
We now define the announced instance $(A_G,\kappa)$.
First of all, $\kappa=\lambda$.
Moreover, for every $i\in \{0,\dots, m-1\}$, the TS $A_G$ has the following path $T_i$, that uses the vertices of $\mathfrak{e}_i=\{v_{i_0}, v_{i_1}\}$ (assuming $i_0<i_1$) as events:
\begin{center}
\begin{tikzpicture}[new set = import nodes]
\begin{scope}[nodes={set=import nodes}]
		\node (T) at (-0.75,0) {$T_i=$};
		\foreach \i in {0,...,2} { \coordinate (\i) at (\i*2cm, 0) ;}
		\foreach \i in {0,...,2} { \node (\i) at (\i) {\nscale{$t_{i,\i}$}};}
\graph {(import nodes);
			0 ->["\escale{$v_{i_0}$}"]1->["\escale{$v_{i_1}$}"]2;
		};
\end{scope}
\end{tikzpicture}
\end{center}

Furthermore, for every $i\in \{0,\dots, n-1\}$, the TS $A$ has the following gadget $F_i$ that uses the node $v_i$ as event, and, for all $j\in \{0,\dots, \kappa\}$, has an $a_j$-labeled edge which directs in the same direction as the $v_i$-labeled edge:
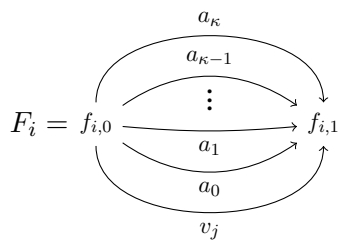
\begin{center}
\begin{tikzpicture}[new set = import nodes]
\begin{scope}[xshift=7.5cm, nodes={set=import nodes}]
	\node (F) at (-0.75, 0) {$F_i=$};
	\coordinate (0) at (0,0);
	\coordinate (1) at (3,0);
	\node (0) at (0) {\nscale{$f_{i,0}$}};
	\node (1) at (1) {\nscale{$f_{i,1}$}};
	\node (dots) at (1.5,0.4) {\nscale{$\vdots$}};
	\graph {
	(import nodes);
			
			0->[ swap, bend right=90,  "\escale{$v_j$}"]1;
			0->[swap, bend right =35, "\escale{$a_0$}"]1;
			0->[ bend right  =5, swap,"\escale{$a_1$}"]1;
			0->[swap, bend left =35,swap , "\escale{$a_{\kappa-1}$}"]1;
			0->[swap, bend left =90, swap, "\escale{$a_\kappa$}"]1;

			};
\end{scope}
\end{tikzpicture}
\end{center}

The TS $A_G$ has the initial state $\iota$;
for all $i\in \{0,\dots, m-1\}$, and all $j\in \{0,1,2\}$, it has the edge $\iota\lEdge{y_i^j}t_{i,j}$;
finally, for all $\ell\in \{0,\dots, n-1\}$, it has the edge $\iota\edge{z_\ell}f_{\ell,0}$.
The $y_i^j$-, and $z_\ell$-labeled edges serve to ensure reachability.
For the sake of simplicity, we summarize these events by $Y=\bigcup_{i=0}^{m-1}\{y_i^0,y_i^1,y_i^2\}\cup\{z_0,\dots, z_{n-1}\}$.

\begin{lemma}\label{lem:edge_removal_implies_model}
If there is an edge-removal $B_G$ of $A_G$ that satisfies $\vert \K\vert \leq \kappa$, and has the $\tau$-SSP, respectively the $\tau$-ESSP, then there is a $\lambda$-VC for $G=(\U,M)$.
\end{lemma}

\begin{proof}
Let $B_G$ be an edge-removal of $A_G$ that satisfies $\vert \K\vert \leq \kappa$, and has the $\tau$-SSP, respectively the $\tau$-ESSP, and let
$\Ss=\U\cap \{e\in E(A_G)\mid \exists s,s'\in S(A_G): s\edge{e}s'\in \K \}$ be the set of events of $A_G$ that label an edge of $A_G$ that is removed to obtain $B_G$.
First of all, we note that $\vert \Ss\vert \leq \vert \K\vert \leq \kappa=\lambda$.
Moreover, in the following, we will argue that $\Ss$ defines a vertex cover of $G$.

\medskip
Let $i\in \{0,\dots, m-1\}$ be arbitrary but fixed.
We show that $\Ss\cap\{v_{i_0}, v_{i_1}\}\not=\emptyset$.
If $t_{i,0}\edge{v_{i_0}}t_{i,1}\in \K$ or $t_{i,1}\edge{v_{i_1}}t_{i,2}\in \K$, then we are finished.
Otherwise, since $B_G$ has the $\tau$-SSP, respectively the $\tau$-ESSP, there is a $\tau$-region $R=(sup, sig)$ that solves $(t_{i,0}, t_{i,2})$, respectively $(v_{i_0}, t_{i,2})$, and thus satisfies $sup(t_{i,0})\not=sup(t_{i,2})$.
This implies either $sup(t_{i,0})=sup(t_{i,1})$, and thus $sig(v_{i_0})\in \{\nop,\free,\used\}$, and $sig(v_{i_1})\in \{\inp, \out, \swap\}$, or $sup(t_{i,1})=sup(t_{i,2})$, and thus $sig(v_{i_0})\in \{\inp,\out, \swap\} $, and $sig(v_{i_1})\in \{\nop,\free,\used\}$.

\medskip
We show that this implies $\K\cap\{f_{i_0}\edge{v_{i_0}}f_{i_0,1}, f_{i_1}\edge{v_{i_1}}f_{i_1,1}\}\not=\emptyset$:
Assume, for a contradiction, that the opposite is true.
Since $\vert \K\vert \leq \kappa$, there is a $j\in \{0,\dots, \kappa\}$, such that both $f_{i_0}\edge{a_j}f_{i_0,1}$, and $f_{i_1}\edge{a_j}f_{i_1,1}$ are present in $B_G$.
If $sig(v_{i_0})\in \{\nop,\free,\used\}$, and $sig(v_{i_1})\in \{\inp, \out, \swap\}$, we have $sup(f_{i_0,0})=sup(f_{i_0,1})$, and $sup(f_{i_1,0}) \not= sup(f_{i_1,1})$, which simultaneously implies $sig(a_j)\in \{\nop,\free,\used\}$, and $sig(a_j)\in \{\inp,\out,\swap\}$, which is a contradiction.
Hence, we have $\K\cap\{f_{i_0}\edge{v_{i_0}}f_{i_0,1}, f_{i_1}\edge{v_{i_1}}f_{i_1,1}\}\not=\emptyset$.
Analogously, if $sig(v_{i_0})\in \{\inp,\out, \swap\} $, and $sig(v_{i_1})\in \{\nop,\free,\used\}$, then we get $\K\cap\{f_{i_0}\edge{v_{i_0}}f_{i_0,1}, f_{i_1}\edge{v_{i_1}}f_{i_1,1}\}\not=\emptyset$ as well.

\medskip
By the arbitrariness of $i$, this implies that $\Ss$ is a vertex cover of $G$.
\end{proof}

For the converse direction, we have to show that the existence of a suitable vertex cover implies that $A_G$ has an implementable edge-removal.
So let $\mS=\{v_{\ell_0},\dots, v_{\ell_{\lambda-1}}\}$ be a vertex cover of $G$, and let $B_G$ be the TS that originates from $A_G$ be removing the edge $f_{\ell_i,0}\edge{v_{\ell_i}}f_{\ell_i,1}$ for all $i\in \{0,\dots, \lambda-1\}$, and nothing else.
One easily verifies that $B_G$ is a well-defined reachable edge-removal of $A_G$ that satisfies $\vert \K\vert \leq \kappa$.

\medskip
In the following, we will show that $B_G$ has the  $\tau$-SSP as well as the $\tau$-ESSP, by presenting regions that altogether solve the individual separation atoms of $B_G$.
Let $I=\{t_{0,0},t_{1,0},\dots, t_{m-1,0}\}\cup\{f_{0,0}, f_{1,0},\dots, f_{n-1,0}\}$ be the set of the initial states of the gadgets of $B_G$, and $E=E(A_G)\setminus Y$ be the set of events in those gadgets (i.e., the vertices of $G$ and the $a_j$'s).

For the sake of simplicity, we often restrict the presentation of a region $R=(sup, sig)$ to the states of $I$ and the events of $E$.
This is justified, since we can easily extend $R$ to $B_G$, when this is possible. Indeed,
choosing as we want $sup(\iota)$, for any $s\in I$ and $\iota\edge{e}s$ (then $e\in Y$ and is unique), we may choose $sig(e)=\nop$ if $sup(\iota)=sup(s)$ and $sig(e)=\swap$ otherwise;
for any $i\in\{0,\ldots,m-1\}$, from the support of $t_{i,0}$, the signatures of $v_{i_0}$ and $v_{i_1}$ determine the supports of $t_{i,1}$ and $t_{i,2}$ (if the functions of the signatures are defined for these supports);
for any $i\in\{0,\ldots,n-1\}$, if $f_{i,0}\edge{e}f_{i,1}$ in $B_G$, from the support of $f_{i,0}$, the signature of $e$ determines the support of $f_{i,1}$ (if defined);
however, if we also have $f_{i,0}\edge{e'}f_{i,1}$ with $e'\neq e$ in $B_G$, it is necessary that the signature of $e'$ is  "compatible" with the pair $sup(f_{i,0})$ and $sup(f_{i,1})$ (to go from $0$ to $0$, $sig(e')$ may be \nop{} or \free; to go from $0$ to $1$, $sig(e')$ may be \swap{} or \out; to go from $1$ to $0$, sig(e') may be \swap{} or \inp; to go from $1$ to $1$, $sig(e')$ may be \nop{} or \used).

\medskip
In fact, the same is true if we choose as we want $sup(\iota)$ as well as $sup(s)$ when one chooses as we want some $s$ in each $T_i$ ($i\in\{0,\ldots,m-1\}$) and each $F_j$  ($j\in\{0,\ldots,n-1\}$), and (coherently) signatures $sig(e)$ when $e\in E$.
This is due to the fact that, for each partial function in $\tau$, its inverse is also a partial function (this would not be true if we allowed \set{} or \res). Hence we may proceed backward as well as forward in the exploration of $B_G$.

\begin{fact}\label{fact:edge_removal_model_implies_ssp}
The TS $B_G$ has the $\tau$-SSP.
\end{fact}

\begin{proof}
-Let $sup_0(\iota)=0$, $sup_0(s)=1$ when $s\in S(A_G)\setminus\{\iota\}$, $sig_0(e)=\swap$ for all $e\in Y$ and $sig_0(e)=\nop$ when $e\in E(A_G)\setminus Y$, then $R_0=(sup_0, sig_0)$ is a region that solves $\iota$.

\medskip
\noindent Let $i\in \{0,\dots, m-1\}$ be arbitrary but fixed. \smallskip\\
- For all $s\in I$, let $sup_1(s)=1$ if $s=t_{i,0}$, $0$ otherwise, and let $sig_1(e)=\nop$ for all $e\in E$.
This leads to a region $R_1=(sup_1, sig_1)$ that solves $(s,s')$ for all $s\in S(T_i)$, and all $s'\in S(A_G)\setminus S(T_i)$.
\smallskip \\
-If $sup_2(s)=0$ for all $s\in I$ and $sig_2(e)=\swap$ for all $e\in E$,
then $R_2=(sup_2, sig_2)$ solves $(t_{i,0}, t_{i,1})$ and $(t_{i,1}, t_{i,2})$.
\smallskip \\
-If $sup_3(s)=0$ for all $s\in I$ and, for all $e\in E$, if $e=v_{i_0}$ then $sig_3(e)=\swap$, $\nop$ otherwise,
then $R_3=(sup_3, sig_3)$ solves $(t_{i,0}, t_{i,2})$.
\smallskip\\
By the arbitrariness of $i$, this shows, for all $s\in \bigcup_{i=0}^{m-1}S(T_i)$, that $s$ is solvable.

\medskip
Similarly, one shows the solvability of each $s\in \bigcup_{i=0}^{n-1}S(F_i)$.
The fact follows.
\end{proof}

In the rest of the subsection, we shall assume that $\omega\subseteq \{\inp,\out, \free,\used\}$ with $\omega\cap\{\inp,\out\}\not=\emptyset$, and $\tau=\{\nop,\swap\}\cup\omega$.

\begin{fact}\label{fact:edge_removal_model_implies_essp}
The TS $B_G$ has the $\tau$-ESSP.
\end{fact}

\begin{proof}
We shall assume that $\inp\in\tau$ (the case where $\omega=\{\out\}$ is symmetrical).\\
-If $sup_0(\iota)=1$, $sig_0(e)=\inp$ for all $e\in Y$ and $sig_0(e)=\nop$ for all $E(A_G)\setminus Y$, then $R_0=(sup_0, sig_0)$ solves $e$ for all $e\in Y$.

\eject 
We proceed with the $a_j$'s:
Let $j\in \{0,\dots, \kappa\}$ be arbitrary but fixed.\\
-If $sup_1(\iota)=0$, $sup_1(s)=1$ for all $s\in I$, $sig_1(a_j)=\inp$ and $sig_1(e)=\swap$ if $e\in E(A_G)\setminus\{a_j\}$,
then $R_1=(sup_1,sig_1)$ solves $(a_j,s)$ for all $s\in \bigcup_{i=0}^{n-1}\{f_{i,1}\}\cup \bigcup_{i=0}^{m-1}\{t_{i,1}\}\cup\{\iota\}$.
\smallskip\\
-If $sup_2(s)=1$ for all $s\in \{f_{0,0},\dots, f_{n-1,0}\}$, $sup_2(s)=0$ for all $s\in \{t_{0,0},\dots, t_{m-1,0}\}$,  $sig_1(a_j)=\inp$ and $sig_1(e)=\swap$ if $e\in E\setminus\{a_j\}$, then $R_2=(sup_2,sig_2)$ solves $(a_j,s)$ for all $s\in \bigcup_{i=0}^{m-1}\{t_{i,0}, t_{i,2}\}$.
\smallskip\\
Since $j$ was arbitrary, the solvability of the $a_j$'s follows.

\medskip
We proceed with the $v_i$'s:
Let $i\in \{0,\dots, m-1\}$ be arbitrary but fixed, and let $L$ select the edges (at most three) that contain $v_i$, that is, $v_i\in \mathfrak{e_\ell}$ when $\ell\in L$.

\medskip
We start with the case where $v_i\in \Ss$, implying that $f_{i,0}\edge{v_i}f_{i,1}\not\in B_G$:
\smallskip\\
-For all $s\in I$, let $sup_3(s)=1$ if $s\in \bigcup_{\ell\in L}\{t_{\ell,0}\}$ and $0$ otherwise; for all $e\in E$, then $sig_3(e)=\inp$ if $e=v_i$ and $\nop$ otherwise.
Then  $R_3=(sup_3,sig_3)$ solves $(v_i,s)$ for all $s\in S\setminus \bigcup_{\ell\in L}S(T_{\ell})$.
\smallskip\\
Let $\ell\in L$.
The region $R_3$ also solves $(v_i,t_{\ell,2})$ and, if $t_{\ell,0}\edge{v_i}$, then it solves also $(v_i,t_{\ell,1})$.
With respect to $T_{\ell}$, it remains to address the case $t_{\ell,1}\edge{v_i}$, which requires to solve $(v_i,t_{\ell,0})$.
The following region $R_4=(sup_4, sig_4)$ does it:
For all $s\in I$, if $s\edge{v_i}$ (implying that $v_i$ is the ``first'' event in the corresponding edge of the graph), then $sup_4(s)=1$, otherwise $sup_4(s)=0$;
for all $e\in E$, if $e=v_i$, then $sig_4(e)=\inp$, otherwise $sig_4(e)=\swap$ (note that, here, it is important that $f_{i,0}\edge{v_i}f_{i,1}\not\in B_G$, otherwise we would lose the coherency in $F_{v_i}$).
Since $i$ and $\ell$ were arbitrary, this shows that solvability of vertex events that do belong to the vertex cover.

\medskip
It remains to consider the case where $v_i\not\in \Ss$, implying that $f_{i,0}\edge{v_i}f_{i,1}\in B_G$:
\smallskip\\
-Let $sup_5(s)=1$ when $s\in I$ and, for all $e\in E$, $sig_5(e)=\inp$ when $e\not\in \Ss$, \nop{} otherwise.
Then $R_5=(sup_5,sig_5)$ solves $(v_i,s)$ for all $s\in \bigcup_{i=0}^{n-1}\{f_{i,1}\}$.
(Notice that never two events with an \inp-signature occur one after the other, since $\Ss$ is a vertex cover.)
\smallskip\\
-If $sup_6(s)=1$ for all $s\in \{f_{i,0}\}\cup\bigcup_{i=0}^{m-1}\{t_{i,0}\}$, $sup_6(s)=0$ for all $s\in\bigcup_{j=0}^{n-1}\{f_{j,0}\}\setminus \{f_{i,0}\}$, and, for all $e\in E$, $sig_6(e)=\inp$ if $e=v_i$, $sig_6(e)=\swap$ if $e\in  \{a_0,\dots, a_{\kappa}\}\cup\U\setminus \Ss$, $sig_6(e)=\nop$ otherwise, then $R_6=(sup_6,sig_6)$ solves $(v_i,s)$ for all $s\in \bigcup_{j=0}^{n-1}\{f_{j,0}\}\setminus \{f_{i,0}\}$.

\medskip
In order to solve $v_i$ within $T_j$ when $j\not\in L$, it suffices to define a region that is complementary to $R_6$ within the $T_j$'s.
The following region $R_7=(sup_7, sig_7)$ accomplishes this:
for all $s\in I$, if $s\in \{ f_{i,0}\}\cup\bigcup_{l\in L}\{t_{l,0}\}$, then $sup_7(s)=1$, otherwise $sup_7(s)=0$;
for all $e\in E$, $sig(e)=\inp$ if $e=v_i$, $sig_7(e)=\swap$ if $e\in \{a_0,\dots, a_{\kappa}\}\cup \U\setminus \Ss$, otherwise $sig_7(e)=\nop$.

\medskip
It remains to solve $v_i$ within $T_j$ when $j\in L$.
Let $j\in L$ be arbitrary but fixed.
The ESSA $(v_i, t_{j,2})$ is solved by $R_5$, and if $t_{j,0}\edge{v_i}$, then $R_5$ solves $(v_i, t_{j,1})$ as well.
Hence, it remains to consider the case $t_{j,1}\edge{v_i}$, which implies $v_{{j}_0}\in \Ss$, and requires to solve $(v_i, t_{j,0})$.
The following region $R_8=(sup_8, sig_8)$ accomplishes this:
for all $s\in I$, if $s\in \{f_{i,0}\}\cup\bigcup_{l\in L}\{t_{l,0}\}\setminus\{t_{j,0}\}$, then $sup_8(s)=1$, otherwise $sup_8(s)=0$;
for all $e\in E$, $sig_8(e)=\inp$ if $e=v_i$, $sig_8(e)=\swap$ if $e\in\{v_{{j}_0}\}\cup\{a_0,\dots, a_{\kappa}\}\cup\U\setminus \Ss$, $sig_8(e)=\nop$ otherwise.
\smallskip\\
Since $i$ and $j$ were arbitrary, this completes the solvability of the vertex events.
The fact follows.
\end{proof}

Altogether, we get the following lemma:
\begin{lemma}\label{lem:vc_implies_edge_removal}
If there is a $\lambda$-VC for $G$, then there is an edge-removal $B_G$ of $A_G$ that satisfies $\vert \K\vert \leq \kappa$,
and has the $\tau$-SSP as well as the $\tau$-ESSP.
\end{lemma}

\subsection{The proof of Theorem~\ref{the:edge_removal}(\ref{the:edge_removal_langsim_real}) for the types without $\inp$
and $\out$}\label{sec:edge_removal_save}%

In this section, we assume that $\emptyset\neq\omega\subseteq \{\free,\used\}$,
and more exactly $\used\in\omega$ (the case $\omega=\{\free\}$ may be handled symmetrically).

\medskip
Note that, like in section~\ref{sec:split-used}, if $A$ is a TS such that $s'\edge{e}s\edge{\neg e}$, then the ESSA $(e,s)$ cannot be solved by a $\tau$-region.
The current reduction then extends the former one by simply adding the missing reverse edges.
More exactly, we first define $\kappa=2\lambda$.
Then, for every $i\in \{0,\dots, m-1\}$, and for every $j\in \{0,\dots, n-1\}$, the new TS $\overline{A}_G$ has the following gadgets $\overline{T}_i$, and $\overline{F}_j$: (where $s\edge{e}s'$ implies $s'\edge{e}s$): 
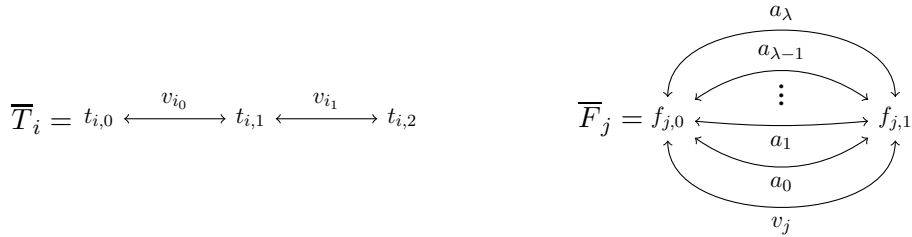
\vspace*{-3mm}

\begin{center}
\begin{tikzpicture}[new set = import nodes]
\begin{scope}[nodes={set=import nodes}]
		\node (T) at (-0.75,0) {$\overline{T}_i=$};
		\foreach \i in {0,...,2} { \coordinate (\i) at (\i*2cm, 0) ;}
		\foreach \i in {0,...,2} { \node (\i) at (\i) {\nscale{$t_{i,\i}$}};}
\graph {(import nodes);
			0 <->["\escale{$v_{i_0}$}"]1<->["\escale{$v_{i_1}$}"]2;
		};
\end{scope}

\begin{scope}[xshift=7.5cm, nodes={set=import nodes}]
	\node (F) at (-0.75, 0) {$\overline{F}_j=$};
	\coordinate (0) at (0,0);
	\coordinate (1) at (3,0);
	\node (0) at (0) {\nscale{$f_{j,0}$}};
	\node (1) at (1) {\nscale{$f_{j,1}$}};
	\node (dots) at (1.5,0.4) {\nscale{$\vdots$}};
	\graph {
	(import nodes);
			
			0<->[ swap, bend right=90,  "\escale{$v_j$}"]1;
			0<->[swap, bend right =35, "\escale{$a_0$}"]1;
			0<->[ bend right  =5, swap,"\escale{$a_1$}"]1;
			0<->[swap, bend left =35,swap , "\escale{$a_{\lambda-1}$}"]1;
			0<->[swap, bend left =90, swap, "\escale{$a_\lambda$}"]1;
			};
\end{scope}
\end{tikzpicture}\vspace*{-1mm}
\end{center}

Again, the TS $\overline{A}_G$ has the initial state $\iota$; for all $i\in \{0,\dots, m-1\}$ and all $j\in \{0,1,2\}$, it has the edges $\iota\FBedge{y_i^j}t_{i,j}$;
finally, for all $\ell\in \{0,\dots, n-1\}$, it has the edges $\iota\FBedge{z_\ell}f_{\ell,0}$.
The result is a TS $\overline{A}_G$, where the $y_i^j$-, and $z_\ell$-labelled edges serve to ensure reachability even if we delete some $v_j$'s.
For the sake of simplicity, we summarize these events by $Y=\bigcup_{i=0}^{m-1}\{y_i^0,y_i^1,y_i^2\}\cup\{z_0,\dots, z_{n-1}\}$.

\begin{lemma}\label{lem:edge_removal_save_essp_implies_model}
If there is an edge-removal $\overline{B}_G$ of $\overline{A}_G$ that satisfies $\vert \K\vert \leq \kappa$, and has the $\tau$-ESSP, then there is a $\lambda$-VC for $G=(\U,M)$.
\end{lemma}
\begin{proof}
Let $\overline{B}_G$ be an edge-removal of $\overline{A}_G$ that satisfies $\vert \K\vert \leq \kappa$, and has the $\tau$-ESSP, and let $\Ss=\{v\in \U\mid s\edge{v}s'\in \K\}$ select the vertex events that label an edge of $\overline{A}_G$ that is removed in $\overline{B}_G$.
Note that $s\edge{e}s'\in \K$ implies $s'\edge{e}s\in \K$, since otherwise we would have an unsolvable ESSA $(e,s')$ in $\overline{B}_G$, which is a contradiction.
This particularly implies $\vert \Ss\vert \leq \frac{\kappa}{2}=\lambda$.

\medskip
We argue that $\Ss$ defines a vertex cover for $G$:\\
Let $i\in \{0,\dots, m-1\}$ be arbitrary, but fixed.
Similarly to the proof of Lemma~\ref{lem:edge_removal_implies_model} one argues that $\Ss\cap\{v_{i_0}, v_{i_1}\}=\emptyset$ implies a contradiction to the solvability of $(v_{i_0}, t_{i,2})$.
Hence, by the arbitrariness of $i$, we have $\Ss\cap \mathfrak{e}_i\not=\emptyset$ for all $i\in \{0,\dots, m-1\}$.
This proves the claim.
\end{proof}

\begin{lemma}\label{lem:edge_removal_save_model_implies_essp}
If there a $\lambda$-VC for $G$, then there is an edge-removal $\overline{B}_G$ of $\overline{A}_G$ that satisfies
$\vert \K\vert \leq \kappa$ and has both the $\tau$-SSP and the $\tau$-ESSP.
\end{lemma}

\begin{proof}
Let $\Ss=\{v_{\ell_0},\dots, v_{\ell_{\lambda-1}}\}$ be a suitable vertex cover, and let $\overline{B}_G$ be the TS that originates from $\overline{A}_G$ by the removal of the edges $f_{\ell_j,0}\Edge{v_{\ell_j}}f_{\ell_j,1}$ and $f_{\ell_j,1}\Edge{v_{\ell_j}}f_{\ell_j,0}$ for all $j\in \{0,\dots, \lambda-1\}$, and nothing else.
Clearly, $\overline{B}_G$ is an edge-removal of $\overline{A}_G$ that satisfies $\vert \K\vert =\kappa=2\lambda$.\\

Note that every region presented for the proof of Lemma~\ref{fact:edge_removal_model_implies_ssp} can be directly transferred to the current edge-removal $\overline{B}_G$, since their signatures only use $\nop$, and $\swap$.
Hence, $\overline{B}_G$ has the $\tau$-SSP.

\medskip
We argue for the $\tau$-ESSP:

\medskip
Let $I=\{t_{0,0}, t_{1,0},\dots, t_{m-1,0}, f_{0,0}, f_{1,0}, f_{n-1,0}\}$ be the set of the initial states of the gadgets of $\overline{A}_G$ (when considered as TS).

Let $i\in \{0,\dots, m-1\}$ be arbitrary but fixed.
In the following, we argue that $y_i^0$, and $y_i^1$, and $y_i^2$ are solvable at the necessary states of $T_i$.
We start with $y_i^0$, and have to distinguish the case where $v_{i_0}$ belongs to the vertex cover, and the case where it does not belong to the vertex cover, which implies that $\overline{F}_{i_0}$ is completely present in $\overline{B}_G$.

\medskip
The following region $R_0=(sup_0, sig_0)$ solves $(y_i^0,t_{i,1})$, and $(y_i^0,t_{i,2})$:

\medskip
The case $v_{i_0}\in \Ss$ (implying that $f_{i_0,0}\fbedge{v_{i_0}}f_{i_0,1}\not\in \overline{B}_G$):
We let $sup(\iota)=sup(s)=1$ for all $s\in I$, and, for all $e\in E(\overline{B}_G)\setminus (\bigcup_{i=0}^{m-1}\{y_i^1,y_i^2\})$,
if $e=y_i^0$, then $sig(e)=\used$,
if $e=v_{i_0}$, then $sig(e)=\swap$, and $sig(e)=\nop$ otherwise.

\medskip
After that we compute the support completely:
for every $j\in \{0,\dots, m-1\}$, and every $\ell\in \{0,1\}$, we let $sup(t_{j,\ell+1})=1-sup(t_{j\ell})$ if $v_{j_\ell}=v_{i_0}$ (implying $sig(v_{j_\ell})=\swap$), and otherwise $sup(t_{j,\ell+1})=sup(t_{j\ell})$;
for every $j\in \{0,\dots, n-1\}$, we let $sup(f_{j,1})=1$ (which is consistent, since $v_{i_0}\in \Ss$).
Finally, we complete the signature: for all $e\in \bigcup_{i=0}^{m-1}\{y_i^1,y_i^2\}$, if $\edge{e}s\in \overline{B}_G$, and $sup(s)=0$, then $sig(e)=\swap$, and $sig(e)=\nop$ otherwise.

\medskip
The case $v_{i_0}\not\in \Ss$ (implying that $f_{i_0,0}\fbedge{v_{i_0}}f_{i_0,1}\in \overline{B}_G$, and $v_{i_1}\in \Ss$, which implies $f_{i_1,0}\fbedge{v_{i_1}}f_{i_1,1}\not\in \overline{B}_G$):
We let $sup(\iota)=sup(s)=1$ for all $s\in I$, and, for all $e\in E(\overline{B}_G)\setminus (\bigcup_{i=0}^{m-1}\{y_i^1,y_i^2\})$,
if $e\in \{y_0^0, y_1^0,\dots, y_{m-1}^0\}\cup\{z_0,\dots, z_{n-1}\}$, then $sig(e)=\used$,
if $e=v_{i_1}$, then $sig(e)=\nop$, and
$sig(e)=\swap$ otherwise.

\medskip
After that we compute the support completely:
for every $j\in \{0,\dots, m-1\}$, and every $\ell\in \{0,1\}$, we let $sup(t_{j,\ell+1})=1-sup(t_{j\ell})$ if $v_{j_\ell}=v_{i_0}$ (implying $sig(v_{j_\ell})=\swap$), and otherwise $sup(t_{j,\ell+1})=sup(t_{j,\ell})$;
for every $j\in \{0,\dots, n-1\}$, we let $sup(f_{j,1})=0$.
Finally, for all $e\in \bigcup_{i=0}^{m-1}\{y_i^1,y_i^2\}$, if $\edge{e}s\in \overline{B}_G$, and $sup(s)=0$, then $sig(e)=\swap$, and $sig(e)=\nop$ otherwise.

\medskip
It is easy to see, that we can define a region $R$ that solves $(y_i^1,t_{i,0})$, and $(y_i^1,t_{i,2})$, respectively $(y_i^2,t_{i,0})$, and $(y_i^2,t_{i,1})$, in a similar way.
We thus refrain from representing the tedious details, and consider $y_i^0, y_i^1, y_i^2$ solved with respect to the necessary states of $\overline{T}_i$.\\

-If $sup_1(\iota)=1$ and,
for all $e\in E({\overline{B}}_G)$,
if $e\in \{y_i^0,y_i^1,y_i^2\}$, then $sig_0(e)=\used$,
if $e\in Y\setminus \{y_i^0,y_i^1,y_i^2\}$ then $sig_1(e)=\swap$, and
$sig_1(e)=\nop$ otherwise,
then the resulting region $R_1=(sup_1, sig_1)$ solves $(y_i^j, s)$ for all $s\in S({\overline{B}}_G) \setminus (\{\iota\}\cup S(\overline{T}_i))$, and all $j\in \{0,1,2\}$.

\medskip\noindent
By the arbitrariness of $i$, this proves the solvability of $y$ for all $y\in \bigcup_{i=0}^{m-1}\{y_i^0,y_i^1,y_i^2\}$.
Similarly, one argues that $z$ is solvable for all $z\in \bigcup_{i=0}^{n-1}\{z_i\}$.

\medskip\smallskip
If $e\in \{a_0,\dots, a_\kappa\}\cup\U$ and $s\in S(\overline{B}_G)$ belongs to a gadget that does not contain $e$, then it is easy to see that $(e, s)$ is $\tau$-solvable:
One simply defines a region that maps exactly to $1$ the states that belong to gadgets of $\overline{B}_G$ that contain $e$ (and $0$ to the others),
maps $e$ to $\used$, choses a suitable signature of the reachability events from $Y$ (either $\nop$ or $\swap$), and maps all the other events $e'$ with $e'\in E(\overline{B}_G)\setminus (\{e\}\cup Y)$ to $\nop$.
This particularly implies the solvability of $a$ for all $a\in \{a_0,\dots, a_\lambda\}$.
Hence, it remains to address the vertex events.\\

Let $i\in \{0,\dots, m-1\}$ be arbitrary but fixed.
By the former discussion, if $f_{i_0,0}\fBedge{v_{i_0}}f_{i_0,1}\in \K$ then $v_{i_0}$ is solvable in $\overline{F}_{i_0}$,
otherwise there is no need to solve it in $\overline{F}_{i_0}$.
Similarly, the solvability of $v_{i_1}$ in $\overline{F}_{i_1}$ needs no further discussion.
Hence, it remains to argue for the solvability of $v_{i_0}$ in $\overline{T}_i$ (the case of $v_{i_1}$ will be similar).
Let us first assume that $v_{i_0}\in\Ss$ (so that it does not occur in $\overline{F}_{i_0}$):
The following region $R_2=(sup_2, sig_2)$ solves $(v_{i_0}, t_{i,2})$ (and thus $v_{i_0}$ in $\overline{T}_i$):
$sup_2(\iota)=0$; and for any $j\in\{0,\ldots,m-1\}$,  \smallskip \\
if $t_{j,0}\Edge{v_{i_0}}t_{j,1}\in \overline{T}_j$, then $sup_2(t_{j,0})=sup_2(t_{j,1})=1$, and $sup_2(t_{j,2})=0$ else\\
if $t_{j,1}\Edge{v_{i_0}}t_{j,2}\in \overline{T}_j$, then $sup_2(t_{j,0})=0$ and $sup_2(t_{j,1})=sup_2(t_{j,2})=1$, otherwise
$sup_2(t_{j,0})=sup_2(t_{j,2})=0$ and $sup_2(t_{j,1})=1$; and,
for every $j\in \{0,\dots,n-1\}$, $sup_2(f_{j,0})=0$ and $sup_2(f_{j,1})=1$; moreover,
for all $e\in E(\overline{B}_G)\setminus Y$, if $e=v_{i_0}$ then $sig_2(e)=\used$, otherwise $sig_2(e)=\swap$;
finally, for all $e\in Y$, if $\iota\edge{e}s$ and $sup(s)=1$, then $sig_2(e)=\swap$, otherwise $sig_2(e)=\nop$.

\medskip
Let us now assume that  $v_{i_0}\not\in\Ss$, so that $v_{i_1}\in\Ss$, since we have a vertex covering (then $v_{i_0}$ occurs in $\overline{F}_{i_0}$, but the vertex cover events do not occur in the $\overline{F}$-gadgets).
For all $e\in E(\overline{B}_G)\setminus Y$, if $e=v_{i_0}$ then $sig_3(e)=\used$, if $e\in \Ss$, then $sig_3(e)=\swap$, otherwise $sig_3(e)=\nop$;
moreover, $sup(\iota)=1$, and
for any $j\in\{0,\ldots,n-1\}$, $sup_3(f_{j,0})=1$, and
for any $j\in\{0,\ldots,m-1\}$, if $v_{j_0}=v_{i_0}$, then $sup_3(t_{j,0})=1$, otherwise $sup_3(t_{j,0})=0$;
the other supports and signatures may then be derived coherently, delivering a region  $R_3=(sup_3, sig_3)$ which, as expected, solves $(v_{i_0}, t_{i,2})$. \smallskip\\
Since $i$ was arbitrary, this completes the proof.
\end{proof}

\section{The complexity of  event-removal}\label{sec:event_removal}%

In this section, we deal with the following modification:
\begin{definition}[Event-Removal]\label{def:event_removal}
Let $A=(S,E,\delta,\iota)$ be an  TS.
A TS\  $B=(S,E',\delta',\iota)$ with $E'\subseteq E$ is an \emph{event-removal} of $A$ if for all $e\in E'$ the following is true:
$s\edge{e}s'\in B$ if and only if $s\edge{e}s'\in A$ for all $s,s'\in S$.
By $\E=E\setminus E'$ we refer to the (set of) removed events.
\end{definition}

We want to stress that an event-removal $B$ is an initialized TS by definition and Remark~\ref{RemInit} systems, i.e., each state remains reachable from the initial one, and each remaining event occurs at least once in $\delta'$ since this was already the case in $A$.

\noindent
\fbox{\begin{minipage}[t][1.8\height][c]{0.88\textwidth}
\begin{decisionproblem}
  \problemtitle{\textsc{$\tau$-Event-Removal for Embedding}}
  \probleminput{A TS\  $A=(S,E,\delta,\iota)$, a natural number $\kappa$.}
  \problemquestion{Does there exist an event-removal $B$ for $A$ by ${\E}$ that has the $\tau$-SSP and satisfies $\vert{\E}\vert\leq \kappa$?}
\end{decisionproblem}
\end{minipage}}
\smallskip

\noindent
\fbox{\begin{minipage}[t][1.8\height][c]{0.88\textwidth}
\begin{decisionproblem}
  \problemtitle{\textsc{$\tau$-Event-Removal for Language-Simulation}}
  \probleminput{A TS\  $A=(S,E,\delta,\iota)$, a natural number $\kappa$.}
  \problemquestion{Does there exist an event-removal $B$ for $A$ by ${\E}$ that has the $\tau$-ESSP and satisfies $\vert{\E}\vert\leq \kappa$?}
\end{decisionproblem}
\end{minipage}}
\smallskip

\noindent
\fbox{\begin{minipage}[t][1.8\height][c]{0.88\textwidth}
\begin{decisionproblem}
  \problemtitle{\textsc{$\tau$-Event-Removal for Realization}}
  \probleminput{A TS\  $A=(S,E,\delta,\iota)$, a natural number $\kappa$.}
  \problemquestion{Does there exist an event-removal $B$ for $A$ by ${\E}$ that has the $\tau$-ESSP and the $\tau$-SSP and satisfies $\vert{\E}\vert\leq \kappa$?}
\end{decisionproblem}
\end{minipage}}

\medskip
The following theorem is the main result of this section:

\begin{theorem}\label{the:event_removal}
If $\omega\subseteq\{\inp,\out, \free,\used\}$, and $\tau=\{\nop,\swap\}\cup\omega$, then
\begin{enumerate}
\itemsep=0.95pt
\item\label{the:event_removal_embedding}
\textsc{$\tau$-Event-Removal for Embedding} is NP-complete.
\item\label{the:event_removal_langsim_real}
\textsc{$\tau$-Event-Removal for Language-Simulation}, and\\
\textsc{$\tau$-Event-Removal for Realization} are NP-complete if $\omega\not=\emptyset$.
\end{enumerate}
\end{theorem}

\subsection{The proof of Theorem~\ref{the:event_removal}(\ref{the:event_removal_embedding}), and the proof of Theorem~\ref{the:event_removal}(\ref{the:event_removal_langsim_real}) for the  types with $\inp$ or $\out$}\label{sec:event_removal_inp_out}%

 Let $\tau=\{\nop,\swap\}\cup\omega$ with $\omega\subseteq \{\inp,\out, \free,\used\}$, and for the input $(G,\lambda)$ of \textsc{3BVC}, where $G=(\U,M)$, let $(A_G,\kappa)$ be defined as in Section~\ref{sec:edge_removal_inp_out}.
In this section, we show that $G$ has a $\lambda$-VC if and only if $A_G$ allows an event-removal $B_G$ that respects $\kappa$, and has the $\tau$-SSP.
Moreover, we also show that if $\tau\cap\{\inp,\out\}\not=\emptyset$, then $G$ has a $\lambda$-VC if and only if $A_G$ allows an event-removal $B_G$ that respects $\kappa$, and has the $\tau$-ESSP.
Altogether, this proves Theorem~\ref{the:event_removal}(\ref{the:event_removal_embedding}), and the statements of Theorem~\ref{the:event_removal}(\ref{the:event_removal_langsim_real}) for the types $\tau\cap\{\inp,\out\}\not=\emptyset$.

\medskip
For the sake of simplicity, in the remainder of this section, if we discuss aspects of the $\tau$-ESSP (where the addressed TS is clear from the context), we always assume that $\tau\cap\{\inp,\out\}\not=\emptyset$.

\begin{lemma}\label{lem:event_removal_implies_model}
If there is an event-removal $B_G$ of $A_G$ that satisfies $\vert \E\vert \leq \kappa$, and has the $\tau$-SSP, respectively the $\tau$-ESSP, then there is a $\lambda$-VC for $G=(\U,M)$.
\end{lemma}
\begin{proof}
Let $B_G$ be an event-removal of $A_G$ that satisfies $\vert \E\vert \leq \kappa$, and has the $\tau$-SSP, respectively the $\tau$-ESSP, and let $\Ss=\E \cap\U$ be the set vertex-events that are removed from $A_G$.\\
Let $i\in \{0,\dots, m-1\}$ be arbitrary but fixed.\\
Since $\vert \E\vert \leq \kappa$, there is a $j\in \{0,\dots, \kappa\}$ such that the event $a_j$, and thus all $a_j$-labeled edges are present in $B_G$.
Hence, similar to the arguments for the proof of Lemma~\ref{lem:edge_removal_implies_model}, one argues that if $\{v_{i_0}, v_{i_1}\}\cap \E=\emptyset$, implying that all $v_{i_0}$-labeled, and $v_{i_1}$-labeled edges are present in $B$, then $(t_{i,0}, t_{i,2})$ is an unsolvable SSA of $B$, respectively $(v_{i_0}, t_{i,2})$ is an unsolvable ESSA of $B$, which is a contradiction.
By the arbitrariness of $i$, and since $\vert \Ss\vert \leq \vert\E\vert\leq \kappa=\lambda$, the claim follows.
\end{proof}

Reversely, if $G$ has a $\lambda$-VC, then $A_G$ allows a suitable event-removal $B_G$ that has the $\tau$-SSP, respectively the $\tau$-ESSP:

\begin{lemma}\label{lem:vc_implies_event_removal}
If there is a $\lambda$-VC for $G=(\U,M)$, then there is an event-removal $B_G$ of $A_G$ that satisfies $\vert \E\vert \leq \kappa$,
 has the $\tau$-SSP, and that has $\tau$-ESSP if $\tau\cap\{\inp,\out\}\not=\emptyset$.
\end{lemma}
\begin{proof}
Let $\Ss=\{v_{\ell_0},\dots, v_{\ell_{\lambda-1}}\}$ be a vertex cover for $G$, and let $B_G$ be the TS that originates from $A_G$ by the removal the events of $\Ss$ (and thus the edges labeled by these events), and nothing else.
By the events of~$Y$, $B_G$ is a reachable TS, which has the same states as $A_G$, and obviously satisfies $\vert \E\vert\leq \kappa=\lambda$.

\medskip
In the following, we argue that the solvability of $B_G$ follows from the regions presented for the proofs of Fact~\ref{fact:edge_removal_model_implies_ssp}, and Fact~\ref{fact:edge_removal_model_implies_essp} of  Section~\ref{sec:edge_removal_inp_out}:

\medskip
Let $B_G'$ be the TS that originates from $A_G$ by removing, for all $i\in \{0,\dots, \lambda-1\}$, the edge $f_{\ell_i,0}\edge{v_{\ell_i}}f_{\ell_i,1}$, and nothing else (that is, $B_G'$ corresponds to the edge-removal defined in Section~\ref{sec:edge_removal_inp_out}).
By definition, both $B_G$, and $B_G'$ have the same states as $A_G$, and thus have the same SSA to solve.
Moreover, if $s\edge{e}s'$ is an edge of $B_G$, then it is an edge of $B_G'$, since, for all $i\in \{0,\dots, \lambda-1\}$, $B_G$ does not only miss the edge $f_{\ell_i,0}\edge{v_{\ell_i}}f_{\ell_i,1}$, but all $v_{\ell_i}$-labeled edges (and nothing else).
In particular, for every event $e\in E(B_G)$, which implies $e\not\in \Ss$, all the (original) $e$-labeled edges (of $A_G$) are present in both $B_G$, and $B_G'$.
Hence, if $(e,s)$ is an ESSA of $B_G$, then it is an ESSA of $B_G'$.
Finally, if $R=(sup, sig)$ is a region of $B_G'$, then its restriction to (the events of) $B_G$ is also a region of $B_G$, since $sup(s)\ledge{sig(e)}sup(s')\in \tau$ is implied for all $s\edge{e}s'\in B_G$, by $s\edge{e}s'\in B_G'$.
Moreover, if $R$ solves a (state or event) separation atom $\alpha$ of $B_G'$ that is also present in $B_G$, then it also solves this atom in $B_G$.
Consequently, since $B_G'$ has the $\tau$-SSP, and the $\tau$-ESSP by the arguments of Fact~\ref{fact:edge_removal_model_implies_ssp}, and Fact~\ref{fact:edge_removal_model_implies_essp}, respectively, we conclude that $B_G$ has also these properties.
\end{proof}

\subsection{The proof of Theorem~\ref{the:event_removal}(\ref{the:event_removal_langsim_real}) for the types without $\inp$ and $\out$}\label{sec:event_removal_save}%

Let $\emptyset\neq\omega\subseteq \{\free,\used\}$, and $\tau=\{\nop,\swap\}\cup\omega$.
In order to complete the proof of Theorem~\ref{the:event_removal}(\ref{the:event_removal_langsim_real}) (for $\tau$), we reduce the input $G=(\U,M)$ to the instance $(\overline{A}_G,\kappa)$, where $\overline{A}_G$ is the TS (with bi-directional edges) as defined in Section~\ref{sec:edge_removal_save}, and $\kappa=\lambda$, and show that $G$ has a $\lambda$-VC if and only if $\overline{A}_G$ allows an event-removal $\overline{B}_G$ that respects $\kappa$, and has the $\tau$-ESSP, respectively the $\tau$-ESSP, and the $\tau$-SSP.
We would like to emphasize that we define $\kappa=\lambda$ here, in contrast to the definition of $\kappa=2\lambda$ in section~\ref{sec:edge_removal_save}.
This (current) definition is justified by the fact that the removal of an event already implies the removal of all its edges.
For details see the following arguments.

\begin{lemma}\label{lem:event_removal_save_essp_implies_model}
If there is an event-removal $\overline{B}_G$ of $\overline{A}_G$ that satisfies $\vert \E\vert \leq \kappa$, and has the $\tau$-ESSP, then there is a $\lambda$-VC for $G=(\U,M)$.
\end{lemma}

\begin{proof}
Let $\overline{B}_G$ be an event-removal $\overline{B}_G$ of $\overline{A}_G$ that satisfies $\vert \E\vert \leq \kappa$, and has the $\tau$-ESSP;
let $\Ss=\U\cap\E$.
Note that  $\vert \Ss\vert \leq \vert \E\vert \leq\kappa=\lambda$.

\medskip
\noindent We argue that $\Ss$ is a vertex cover: \smallskip\\
Let $i\in \{0,\dots, m-1\}$ be arbitrary but fixed.
If $\{v_{i_0},v_{i_1}\}\cap\Ss=\emptyset$, then $\overline{T}_i$, $\overline{F}_{i_0}$ and $\overline{F}_{i_1}$ are completely present in $\overline{B}_G$.
Similar to the proof of Lemma~\ref{lem:edge_removal_save_essp_implies_model} one argues that this contradicts the $\tau$-ESSP of $\overline{B}_G$.
Hence $\{v_{i_0},v_{i_1}\}\cap\Ss\not=\emptyset$, and the arbitrariness of $i$ implies the claim.
\end{proof}

\begin{lemma}\label{lem:event_removal_save_model_implies_essp}
If there is a $\lambda$-VC for $G=(\U,M)$, then there is an event-removal $\overline{B}_G$ of $\overline{A}_G$
that satisfies $\vert \E\vert \leq \kappa$ and has the $\tau$-ESSP as well as the $\tau$-SSP.
\end{lemma}

\begin{proof}
Let $\Ss=\{v_{\ell_0},\dots, v_{\ell_{\lambda-1}}\}$ be a suitable vertex cover of $G$, and let $\overline{B}_G$ be the event-removal that originates from $\overline{A}_G$ by removing the events of $\Ss$ (and the corresponding edges), and nothing else.
One easily checks that $\overline{B}_G$ is a (well-defined) reachable TS, and satisfies $\vert \E\vert\leq \vert \Ss\vert$.

\medskip
In the following, we argue that the $\tau$-SSP, and the $\tau$-ESSP of $\overline{B}_G$ is implied by the regions presented for the proof of Lemma~\ref{lem:edge_removal_save_model_implies_essp}:
Let $\overline{B}_G'$ be the TS that originates from $\overline{A}_G$ by the removal of the edges $f_{\ell_j,0}\edge{v_{\ell_j}}f_{\ell_j,1}$, and $f_{\ell_j,1}\edge{v_{\ell_j}}f_{\ell_j,0}$, and nothing else (that is, $\overline{B}_G'$ corresponds to the edge-removal of Section~\ref{sec:edge_removal_save}).
One verifies that $\overline{B}_G$, and $\overline{B}_G'$ have the same states, and thus the same SSAs to solve.
Moreover, if $s\edge{e}s'$ is an edge of $\overline{B}_G$, implying that $e\not\in \Ss$, then $s\edge{e}s'\in \overline{B}_G'$.
Hence, every ESSA of $\overline{B}_G$ is an ESSA of $\overline{B}_G'$.
Moreover, for every region $R=(sup, sig)$ of $\overline{B}_G'$, the restriction of $R$ to (the events of) $\overline{B}_G$ is a region of $\overline{B}_G$, since $sup(s)\edge{e}sup(s')$ is implied for every $s\edge{e}s'$ of $\overline{B}_G$ (by $s\edge{e}s'\in \overline{B}_G'$, and the region property of $R$).
Hence, since the regions of Lemma~\ref{lem:edge_removal_save_model_implies_essp} justify the $\tau$-SSP, and the $\tau$-ESSP, they particularly justify these properties for $\overline{B}_G$.
\end{proof}

\section{The complexity of state-removal}\label{sec:state_removal}%

In this section, we address the following modification:

\begin{definition}[State-Removal]\label{def:state_removal}
Let $A=(S,E,\delta,\iota)$ be a TS.
A TS $B=(S',E,\delta',\iota)$ with states $S'\subseteq S$ is a \emph{state-removal} of $A$ if the following two
conditions are satisfied: \smallskip\\
(1) for all $e\in E$
and all $s,s'\in S'$, $s\edge{e}s'\in B$ if and only if $s\edge{e}s'\in A$; \smallskip \\
(2) if $s\edge{e}s'\in A$ and $s\edge{e}s'\not\in B$, then $s\not\in S'$ or $s'\not\in S'$ (or both). \medskip\\
By ${\mS}=S\setminus S'$ we refer to the (set of) removed states.
\end{definition}
In the following, we shall assume that $B$ is a valid system, i.e., each state remains reachable from the initial one and each event occurs at least once in $\delta'$.

\medskip
In particular, we investigate the complexity of the following three decision problems:

\noindent
\fbox{\begin{minipage}[t][1.8\height][c]{0.88\textwidth}
\begin{decisionproblem}
  \problemtitle{\textsc{$\tau$-State-Removal for Embedding}}
  \probleminput{A TS\  $A=(S,E,\delta,\iota)$, a natural number $\kappa$.}
  \problemquestion{Does there exist a state-removal $B$ for $A$ by ${\mS}$ that has the $\tau$-SSP and satisfies $\vert{\mS}\vert\leq \kappa$?}
\end{decisionproblem}
\end{minipage}}
\smallskip

\noindent
\fbox{\begin{minipage}[t][1.8\height][c]{0.88\textwidth}
\begin{decisionproblem}
  \problemtitle{\textsc{$\tau$-State-Removal for Language-Simulation}}
  \probleminput{A TS\  $A=(S,E,\delta,\iota)$, a natural number $\kappa$.}
  \problemquestion{Does there exist a state-removal $B$ for $A$ by ${\mS}$ that has the $\tau$-ESSP and satisfies $\vert{\mS}\vert\leq \kappa$?}
\end{decisionproblem}
\end{minipage}}
\smallskip

\noindent
\fbox{\begin{minipage}[t][1.8\height][c]{0.88\textwidth}
\begin{decisionproblem}
  \problemtitle{\textsc{$\tau$-State-Removal for Realization}}
  \probleminput{A TS\  $A=(S,E,\delta,\iota)$, a natural number $\kappa$.}
  \problemquestion{Does there exist a state-removal $B$ for $A$ by ${\mS}$ that has the $\tau$-ESSP and the $\tau$-SSP and satisfies $\vert{\mS}\vert\leq \kappa$?}
\end{decisionproblem}
\end{minipage}}

\medskip

The following theorem is the main result of this section:

\begin{theorem}\label{the:state_removal}
If $\omega\subseteq\{\inp,\out, \free,\used\}$, and $\tau=\{\nop,\swap\}\cup\omega$, then
\begin{enumerate}
\itemsep=0.95pt
\item\label{the:state_removal_embedding}
\textsc{$\tau$-State-Removal for Embedding} is NP-complete.
\item\label{the:state_removal_langsim_real}
\textsc{$\tau$-State-Removal for Language-Simulation}, and\\
\textsc{$\tau$-State-Removal for Realization} are NP-complete if $\omega\not=\emptyset$, otherwise they are solvable in polynomial time.
\end{enumerate}
\end{theorem}

First of all, we argue for the polynomial part:
If $\tau=\{\nop,\swap\}$, then a TS $A=(S,E,\delta,\iota)$ has the $\tau$-ESSP if and only if every event occurs at every state,
since the functions  \nop,\swap{} are defined on both $0$ and $1$.
Thus, any ESSP atom $(e,s)$ of $A$ would be unsolvable.
Again, we may determine the states $s\in S$ such that, for some event $e$, we have $\neg s\edge{e}$: they must be removed.
Let $B=(S',E,\delta, \iota)$ be the (unique) result of this phase.
For the language-simulation problem, this is enough and we simply have to check if $B$ is valid and if the number $k$ of removed states does not exceed $\kappa$.
Let $(s,s')$ be an unsolvable SSA of $B$.
Since $B$ is an initialized TS, there is a path $\iota\edge{e_1}s_1\dots s_{n-1}\edge{e_n}s_n$ such that $s_n=s$ in $B$.
If we remove $s_n$ (and get, say, $B'$), then $\neg s_{n-1}\edge{e_n}$, since $\nop$ and $\swap$ are functions, and thus we have an unsolvable ESSA $(s_{n_1}, e_n)$.
Consequently, we have also to remove $s_{n-1}$, and get, say, $B''$.
By the same arguments, we inductively obtain that $s_{n-1},\dots,s_1$, and finally $\iota$, have to be removed, which is a contradiction, since every state-removal of $A$ has the initial state $\iota$ by Definition~\ref{def:state_removal}.
Similarly, the removal of $s'$ yields also a contradiction.

\vspace*{2mm}
Thus, for the proof of Theorem~\ref{the:state_removal}, it remains to consider the NP-completeness results.

\subsection{The proof of Theorem~\ref{the:state_removal}(\ref{the:state_removal_embedding}), and the proof of Theorem~\ref{the:state_removal}(\ref{the:state_removal_langsim_real}) for the types with $\inp$ or $\out$}\label{sec:state_removal_inp_out}%

Let $\omega\subseteq \{\inp,\out, \free,\used\}$ and $\tau=\{\nop,\swap\}\cup\omega$ be such that $\tau\cap\{\inp,\out\}\not=\emptyset$.

\medskip
In order to prove Theorem~\ref{the:state_removal}(\ref{the:state_removal_embedding}), and the statement of Theorem~\ref{the:state_removal}(\ref{the:state_removal_langsim_real}) for $\tau\cap\{\inp,\out\}\not=\emptyset$, we use the following reduction that transforms an input $(G,\lambda)$, where $G=(\U,M)$, to an instance $(A_G,\kappa)$:
$\kappa=\lambda$, and $A_G$ is the TS that is defined just like the one in Section~\ref{sec:edge_removal_inp_out}, but, for all $i\in \{0,\dots, m-1\}$, and for all $j\in \{1,2\}$, it does not implement the edges $\iota\edge{y_i^j}t_{i,j}$.
(The edge $\iota\edge{y_i^0}t_{i,0}$ is still present to ensure reachability.)
By doing so, we ensure that if $B_G$ is well-defined state-removal of $A_G$, then, for all $i\in \{0,\dots, m-1\}$, if $t_{i,1}$ is removed, then $t_{i,2}$ is also removed, and, moreover, if $t_{i,0}$ is removed, then also both $t_{i,1}$ and $t_{i,2}$ are removed.
This can be seen as follows:
By Definition~\ref{def:state_removal}, $B_G$ is a TS, and thus is initialized; hence the initial state $\iota$ is present and its states are reachable from $\iota$ by a directed path.
Hence, if, for example, the state $t_{0,1}$ is missing in $B_G$ (compared with $A_G$), implying that the (only) edge $t_{0,1}\edge{v_{0_1}}t_{0,2}$ is removed too, then the state $t_{0,2}$ is also not present in $B_G$, since it would not be reachable by a directed path from $\iota$ otherwise, which is a contradiction.

\begin{lemma}\label{lem:ssp_state_removal_implies_model}
If there is a state-removal $B_G$ of $A_G$ that satisfies $\vert \mS\vert \leq \kappa$, and has the $\tau$-SSP, respectively
the $\tau$-ESSP with $\tau\cap\{\inp,\out\}\not=\emptyset$, then there is a $\lambda$-VC for $G=(\U,M)$.
\end{lemma}
\begin{proof}
Let $B_G$ be a state-removal of $A_G$ that satisfies $\vert \mS\vert \leq \kappa$, and has the $\tau$-SSP, respectively the $\tau$-ESSP.
Let $\Ss=\{v\in \U\mid \exists s\in \mS : \edge{v}s\}$.
Note that $\vert \Ss\vert \leq \vert \mS\vert$, since every state of $A$ is the target of at most one vertex event.
We show that $\Ss$ defines a vertex cover:
Let $i\in \{0,\dots, m-1\}$ be arbitrary but fixed.
Recall that, as argued above, $\iota$ is present in $B_G$, and, for every path $s_0\edge{e_1}s_1\dots\edge{e_n}s_n$ of $A_G$, if $s_i$ is missing in $B_G$, then $s_j$ is also missing for all $i < j\in \{0,\dots, n\}$, since $B_G$ would have unreachable states otherwise.
Hence, if $\Ss\cap\{v_{i_0}, v_{i_1}\}=\emptyset$, implying $\{t_{i,1}, t_{i,2}, f_{i_0,1}, f_{i_1,1}\}\subseteq S(B_G)$, then $T_i$, and $F_{i_0}$, and $F_{i_1}$ are completely present in $B_G$.
Consequently, similar to the proof of Lemma~\ref{lem:edge_removal_implies_model}, one argues that this implies that $B_G$ has the unsolvable SSA $(t_{i,0}, t_{i,2})$, respectively the unsolvable ESSA $(v_{i_0}, t_{i,2})$, which is a contradiction.
By the arbitrariness of $i$, $\Ss$ is a suitable vertex cover, which proves the lemma.
\end{proof}

For the converse direction, we state the following lemma:

\begin{lemma}\label{lem:ssp_state_removal_implied_model}
If there is a $\lambda$-VC for $G=(\U,M)$, then there is a state-removal $B_G$ of $A_G$ that satisfies $\vert \mS\vert \leq \kappa$ and has the $\tau$-SSP, and, since $\tau\cap\{\inp,\out\}\not=\emptyset$, it has the $\tau$-ESSP.
\end{lemma}
\begin{proof}
Let $\Ss=\{v_{\ell_0}, \dots, v_{\ell_{\lambda-1}}\}$ be a vertex cover for $G$, and let $B_G$ the TS that originates from $A_G$ by removing the states $f_{\ell_0,1}, f_{\ell_1,1},\dots, f_{\ell_{\lambda-1},1}$ (and thus the edges incident to these states), and nothing else.
Notice that, for all $j\in \{0,\dots, \lambda-1\}$, this reduces $F_{\ell_j}$ to $f_{\ell_j,0}$, since, for all $a\in \{a_0,\dots, a_{\lambda-1}\}$, the edge $f_{\ell_j,0}\edge{a}f_{\ell_j,1}$ is removed by the removal of $f_{\ell_j,1}$.
Obviously, $\vert \mS\vert\leq \kappa=\lambda$.

\medskip
Let $B_G'$ be the TS that result by the removal of the edge $f_{\ell_j,0}\Edge{v_{\ell_j}}f_{\ell_j,1}$ for all $j\in \{0,\dots, \lambda-1\}$, that is, $B_G'$ corresponds to the edge-removal defined in Section~\ref{sec:edge_removal_inp_out}.
Obviously, every edge of $B_G$ is present in $B_G'$.
Hence, every region $R=(sup, sig)$ of $B_G'$ is a region of $B_G$, when restricted to its present states, and events.
Furthermore, if $(s,s')$ is an SSA of $B_G$, then it is also an SSA of $B_G'$, and solved by a corresponding region of Fact~\ref{fact:edge_removal_model_implies_ssp}.

If $(e,s)$ is an ESSA of $B_G$ such that $e\not\in \{a_0,\dots, a_\lambda\}$, or $s\not\in \{f_{\ell_0},\dots, f_{\ell_{\lambda-1}}\}$, then $(e,s)$ is also an ESSA of $B_G'$, and solved by a corresponding region of Fact~\ref{fact:edge_removal_model_implies_essp}.

\medskip
Hence, it only remains to argue that an ESSA $(e,s)$ is also solvable, if $e\in \{a_0,\dots, a_\lambda\}$ and $s\in \{f_{\ell_0},\dots, f_{\ell_{\lambda-1}}\}$.
The following region $R=(sup, sig)$ accomplishes this:
$sup(\iota)=0$, and, for all $e\in E(B_G)$, if $e\in \{a_0,\dots, a_\lambda\}$, then $sig(e)=\inp$, and if $e\in \{z_{\ell_0},\dots, z_{\ell{\lambda-1}}\}$, then $sig(e)=\nop$, otherwise $sig(e)=\swap$.
Notice that this definition implies $sup(f_{\ell_j,0})=0$ for all $j\in \{0,\dots,\lambda-1\}$ (by $sig(z_{\ell_j})=\nop$), and $sup(f_{i,1})=sup(f_{i,0}) - 1= 0$ for all $i\in \{0,\dots, n-1\}\setminus \{\ell_0,\dots, \ell_{\lambda-1}\}$ (by $sig(z_i)=sig(v_i)=\swap$, respectively $sig(a_0)=\dots=sig(a_{\lambda})=\inp$) such that $R$ is actually a well-defined region that solves the remaining ESSA.
This proves the lemma.
\end{proof}

\subsection{Proof of Theorem~\ref{the:state_removal}(\ref{the:state_removal_langsim_real}) for the Types without $\inp$ and $\out$}\label{sec:state_removal_save}%

Let $\emptyset\neq\omega\subseteq \{\free,\used\}$, and $\tau=\{\nop,\swap\}\cup\omega$.

\medskip
The following reduction transforms an input $(G,\lambda)$, where $G=(\U,M)$, to an instance $(\overline{A}_G,\kappa)$:
$\kappa=\lambda$, and $\overline{A}_G$ is the TS that is defined just like the one in Section~\ref{sec:edge_removal_save}, but, for all $i\in \{0,\dots, m-1\}$, and for all $j\in \{1,2\}$, does neither implement the edge $\iota\edge{y_i^j}t_{i,j}$ nor the edge $\iota\edge{y_i^j}t_{i,j}$.
Hence, for every state-removal $\overline{B}_G$ of $\overline{A}_G$, if there is $i\in \{0,\dots, m-1\}$, such that $t_{i,1}$ is missing in $\overline{B}_G$ (implying that the edge $t_{i,1}\edge{v_{i_1}}t_{i,2}$ is also removed), then $t_{i,2}$ is also removed, since this state would not be reachable from the initial state $\iota$ otherwise, which is a contradiction.
Similarly, if $t_{i,0}$ is removed, then so are $t_{i,1}$, and $t_{i,2}$.

\medskip
In the following, we argue that $(\overline{A}_G,\kappa)$ is a yes-instance (i.e., it allows a fitting state-removal that has the ESSP, and, when realization is considered, the SSP) if and only if $(G,\lambda)$ is a yes-instance.

\begin{lemma}\label{lem:state_removal_save_essp_implies_model}
If there is a state-removal $\overline{B}_G$ of $\overline{A}_G$ that satisfies $\vert \mS\vert \leq \kappa$ and has the $\tau$-ESSP, then there is a $\lambda$-VC for $G=(\U,M)$.
\end{lemma}

\begin{proof}
Let $\overline{B}_G$ be a state-removal of $\overline{A}_G$ that satisfies $\vert \mS\vert \leq \kappa$ and has the $\tau$-ESSP, and
let $\Ss=\{v\in \U \mid \exists s\in \mS: \edge{v}s\}$.
Note that  $\vert \Ss\vert \leq \vert \mS\vert \leq\kappa=\lambda$, which can be seen as follows:
On the one hand, if $t_{i,1}\in \mS$ for some $i\in \{0,\dots, m-1\}$, which implies $v_{i_0}, v_{i_1}\in \Ss$, then $t_{i,2}\in \mS$.
(Otherwise, $t_{i,2}$ would no reachable, which is excluded by the definition of TS.)
On the other hand, there is no other kind of state that is adjacent to two different vertex events.

\medskip
\noindent We argue that $\Ss$ is a vertex cover:\\
Let $i\in \{0,\dots, m-1\}$ be arbitrary but fixed.
If $\{v_{i_0},v_{i_1}\}\cap\Ss=\emptyset$, then $\overline{T}_i$, and $\overline{F}_{i_0}$, and $\overline{F}_{i_1}$ are completely present in $\overline{B}_G$.
Similar to the proof of Lemma~\ref{lem:edge_removal_save_essp_implies_model} one argues that this contradicts the $\tau$-ESSP of $\overline{B}_G$.
Hence $\{v_{i_0},v_{i_1}\}\cap\Ss\not=\emptyset$, and the arbitrariness of $i$ implies the claim.
\end{proof}

For the converse direction, we present the following lemma:
\begin{lemma}\label{lem:state_removal_save_model_implies_essp}
If there is a $\lambda$-VC for $G=(\U,M)$, then there is a state-removal $\overline{B}_G$ of $\overline{A}_G$ that satisfies $\vert \mS\vert \leq \kappa$ and has the $\tau$-ESSP, and the $\tau$-SSP.
\end{lemma}

\begin{proof}
Let $\Ss=\{v_{\ell_0},\dots, v_{\ell_{\lambda-1}}\}$ be a vertex cover of $G$, and let $\overline{B}_G$ be the TS that originates from $\overline{A}_G$ by the removal of the states $f_{\ell_0,1}, f_{\ell_1,1}, \dots, f_{\ell_{\lambda-1},1}$ (and the corresponding adjacent edges), and nothing else.
One finds our that $\overline{B}_G$ is a well-defined (i.e. reachable) TS, which satisfies $\vert \mS\vert \leq \kappa=\lambda$.
Similar to the proof of Lemma~\ref{lem:event_removal_save_model_implies_essp}, one argues that $\overline{B}_G$ has both the $\tau$-SSP, and the $\tau$-ESSP. The claim \linebreak  follows.
\end{proof}

\section{Concluding  remarks}\label{sec:con}%

In this paper, we characterized the computational complexity of finding a label-splitting of a TS $A$ that allows implementing $\tau$-net for all types $\tau=\{\nop,\swap\}\cup \omega$ with $\omega\subseteq \{\inp,\out,\used,\free\}$ and all implementations previously considered in the literature.
By the results of~\cite{topnoc/Tredup21,
tamc/TredupR19,apn/TredupR19}, the synthesis problem aiming at language-simulation and realization is NP-complete for all types $\tau=\{\nop,\swap\}\cup\omega$ with $\omega\subseteq \{\inp,\out,\res,\set,\used,\free\}$ and $\omega\cap\{\inp,\out,\used,\free\}\neq\emptyset$ and $\omega\cap\{\set,\res\}\neq\emptyset$.
Hence, their corresponding label-splitting problem is also NP-complete.
Moreover, similar to the proof of Theorem~\ref{the:label_splitting}, one argues that label-splitting aiming at language-simulation or realization is trivial for $\tau=\{\nop,\swap\}\cup\omega$ when $\omega\subseteq \{\res,\set\}$:
the relabeling of an input TS $A$ would result in some ESSAs, and since $\tau$ does not allow for solving ESSAs, either $A$ already has the separation properties or it has to be rejected.
Moreover, we already know that synthesis aiming at embedding is NP-complete for all Boolean types $\tau$ with $\{\nop,\swap\}\subseteq\tau$ and $\tau\cap\{\res,\set\}=\emptyset$~\cite{ictac/TredupE20}.
Hence, again their corresponding label-splitting problem is also NP-complete.
Altogether, with the present work, the complexity of $\tau$-label-splitting is characterized for all 64 Boolean types with $\{\nop,\swap\}\subseteq\tau$ and all implementations.
It remains future work to determine the complexity of the label-splitting problem for the \swap-free and \nop-equipped Boolean types whose underlying synthesis problem is~\mbox{in~ P.}

 \medskip
In order to avoid to deal with the intractability of label-splitting, we also analyzed the complexity of various element removal techniques to render a $\tau$-implementable TS, for similar types $\tau\cup\omega$ with $\emptyset\not=\omega\subseteq\{\inp,\out, \used,\free\}$.
Unfortunately, it turns out that these techniques are also NP-complete for (almost) all $\tau$'s, and all implementations.
There is however a single case that, surprisingly, remains undetermined:
the complexity for event-removal when $\omega=\emptyset$.
Hence, it remains for future work to investigate this special case, and to search for other techniques that may have a better complexity to make a TS implementable, like for example the insertion of states, edges or \linebreak events.

\medskip
Moreover, one might consider $\kappa$ as a parameter and investigate the discussed modification techniques from the point of view of parameterized complexity.

\subsubsection*{Acknowledgements.}
We would like to thank the unknown reviewers of this paper and the previous versions of this paper for their valuable comments.

\bibliographystyle{fundam}
\bibliography{myBibliography}

\end{document}